\documentclass[11pt]{article}

\usepackage{latexsym, amsmath, amssymb, amsthm}

\newif\iffull\fulltrue

\newtheorem{theorem}{\bf Theorem}[section]
\newtheorem{lemma}[theorem]{\bf Lemma}
\newtheorem{corollary}[theorem]{\bf Corollary}
\newtheorem{proposition}[theorem]{\bf Proposition}

\newtheorem{fact}{\bf Fact}

\def\rk{{\rm rk}}
\def\chara{{\rm char}}
\def\Z{{\mathbb Z}}
\def\Q{{\mathbb Q}}

\def\F{{\mathbb F}}
\def\qed{\raisebox{-.2ex}{$\Box$}}
\def\Aalg{{\cal A}}
\def\Balg{{\cal B}}
\def\Calg{{\cal C}}
\def\Dalg{{\cal D}}

\def\Iid{{\cal I}}

\def\forevery{{\mbox{ for every }}}
\def\Aut{{\rm Aut}}
\def\rem#1{}

\begin{document}

\date{}
\title{
\bf Trading GRH for Algebra: Algorithms for Factoring Polynomials and
Related Structures
}
\author{
G\'abor Ivanyos
\thanks{Computer and Automation Research Institute
of the Hungarian Academy of Sciences (MTA SZTAKI),
L\'agym\'anyosi u. 11,
1111 Budapest, Hungary. 
E-mail: {\tt Gabor.Ivanyos@sztaki.hu}}
\and
Marek Karpinski
\thanks{Department of Computer Science and Hausdorff Center for Mathematics, 
University of Bonn, 53117 Bonn, Germany.
E-mail: {\tt marek@cs.uni-bonn.de}}
\and
Lajos R\'onyai
\thanks{
MTA SZTAKI
and 
Department of Algebra, 
Budapest University of Technology and Economics,
M\H{u}egyetem rkp.~3-9,
1111 Budapest, Hungary.
E-mail: {\tt lajos@ilab.sztaki.hu}}
\and 
Nitin Saxena
\thanks{Hausdorff Center for Mathematics,
Endenicher Allee 62,
53115 Bonn, Germany.
E-mail: {\tt ns@hcm.uni-bonn.de}}
\thanks{
The authors thank the Hausdorff Research Institute for Mathematics for its kind support.}
}

\maketitle
\begin{abstract}
In this paper we develop techniques that eliminate the need of the Generalized 
Riemann Hypothesis (GRH) from various (almost all) known results about deterministic 
polynomial factoring over finite fields. Our main result shows that given a 
polynomial $f(x)$ of degree $n$ over a finite field $k$, we can find in deterministic 
$poly(n^{\log n}$, $\log |k|)$ time {\em either} a nontrivial factor 
of $f(x)$ {\em or} a nontrivial automorphism of $k[x]/(f(x))$ of order $n$. 
This main tool leads to various new GRH-free results, most striking of which are:
\begin{enumerate}
\item Given a noncommutative algebra $\Aalg$ of dimension $n$ over a finite field 
$k$. There is a deterministic $poly(n^{\log n}$, $\log |k|)$ time algorithm to find
a zero divisor in $\Aalg$. This is the best known deterministic GRH-free result 
since Friedl and R\'onyai (STOC 1985) first studied the problem of finding zero divisors in finite 
algebras and showed that this problem has the same complexity as factoring 
polynomials over finite fields.
\item Given a positive integer $r$ such that either $8|r$ or 
$r$ has at least two distinct odd prime factors. There is a deterministic polynomial time 
algorithm to find a nontrivial factor of the $r$-th cyclotomic 
polynomial over a finite field. This is the best known deterministic GRH-free result 
since Huang (STOC 1985) showed that cyclotomic polynomials can be factored 
over finite fields in deterministic polynomial time assuming GRH.
\end{enumerate}
In this paper, following the seminal work of Lenstra (1991)
on constructing isomorphisms between finite fields,
we further generalize classical Galois
theory
constructs like cyclotomic extensions, Kummer extensions,
Teichm\"uller subgroups, to the case of commutative 
semisimple algebras with automorphisms. These
generalized constructs help eliminate the dependence on GRH.
\end{abstract}

\section{Introduction}

The problem of finding a nontrivial factor of a given polynomial over a finite 
field is a fundamental computational problem. There are many problems whose known
algorithms first require factoring polynomials. Thus, polynomial factoring is an
intensely studied question and various randomized polynomial time algorithms are
known -- Berlekamp \cite{Be}, Rabin \cite{rabin}, Cantor and Zassenhaus \cite{cantor-z}, von 
zur Gathen and Shoup \cite{gathen-shoup}, Kaltofen and Shoup \cite{kaltofen-shoup} 
-- but its deterministic complexity is a longstanding open problem. There are
although several partial results known about the deterministic complexity of 
polynomial factoring based on the conjectured truth of the generalized Riemann 
Hypothesis (GRH). The surprising connection of GRH with polynomial factoring is 
based on the fact that if GRH is true 
and $r$ is a prime dividing $(|k|-1)$ then
one can find primitive $r$-th 
nonresidues in the finite field $k$, which can then be used to factor `special' 
polynomials, $x^r-a$ over $k$, in deterministic polynomial time (see 
\cite{ev-solvable}). 

Based on this are many 
deterministic factoring algorithms known, but all of them are super-polynomial time 
except on special instances. 

The special instance when the degree $n$ of the input polynomial $f(x)$ has a 
``small'' prime factor $r$ has been particularly interesting. R\'onyai \cite{ro1} 
showed that under GRH one can find a nontrivial factor of $f(x)$ in deterministic 
polynomial time. Later it was shown by Evdokimov \cite{Ev} that R\'onyai's 
algorithm can be modified to get under GRH a deterministic algorithm that factors 
{\em any} input polynomial $f(x)\in k[x]$ of degree $n$ in {\em sub-exponential} time 
$poly(n^{\log n}, \log|k|)$. This line of approach has since been investigated, in
an attempt to remove GRH or improve the time complexity, leading to several 
algebraic-combinatorial conjectures and quite special case solutions 
\cite{ch, gao, iks}. 

Some other instances studied have been related to the {\em Galois group} of the 
given polynomial over rationals. R\'onyai \cite{ro4} showed under GRH that any
polynomial $f(x)\in\Z[x]$ can be factored modulo $p$ deterministically in time 
polynomial in the size of the Galois group over $\Q$ of $f$, except for finitely 
many primes $p$. 
Other results of a similar flavor are: Evdokimov \cite{ev-solvable} showed under 
GRH that $f(x)$ can be factored in deterministic polynomial time if it has a 
{\em solvable} Galois group while Huang \cite{huang} showed under GRH that $f(x)$ 
can be factored in deterministic polynomial time if it has an {\em Abelian} Galois 
group. 

Another instance studied is that of ``special'' finite fields. Bach, von zur 
Gathen and Lenstra \cite{bgl} showed under GRH that polynomials over finite fields 
of characteristic $p$ can be factored in deterministic polynomial time if 
$\phi_k(p)$ is ``smooth'' for some integer $k$, where $\phi_k(x)$ is the $k$-th 
cyclotomic polynomial. This result generalizes the previous works of R\'onyai 
\cite{ro2}, Mignotte and Schnorr \cite{mignotte}, von zur Gathen \cite{gathen}, 
Camion \cite{camion} and Moenck \cite{moenck}.

Polynomial factoring has several applications both in the real world - coding 
theory and cryptography - and in fundamental computational algebra problems. The 
latter kind of applications are relevant to this work. Friedl and R\'onyai \cite{FR} studied
the computational problem of finding the 
simple components and a zero 
divisor of a given finite algebra over a finite field. They showed that all these 
problems depend on factoring polynomials over finite fields and hence have 
randomized polynomial time algorithms. Furthermore, they have under GRH 
deterministic subexponential time algorithms. In this work we give an unconditional 
version of this result. We show that if the given algebra is noncommutative then 
in fact we can find a zero divisor in deterministic subexponential time {\em 
without} needing GRH.

\subsection{Our Results and Techniques} 

As we saw above there are several results on polynomial factoring that assume the 
truth of the GRH. Of course one would like to eliminate the need of GRH but that 
goal is still elusive. As a first step in that direction we give in this work GRH 
free versions of all the results mentioned above. In these versions the basic tool 
is that we either successfully find a nontrivial factor of a polynomial $f(x)$ 
over a finite field $k$ {\em or} we find a nontrivial automorphism of the algebra 
$k[x]/(f(x))$. Formally speaking the main result of the paper is:
\\

\noindent
{\bf Main Theorem: }
{\em Let $\Aalg$ be a commutative semisimple algebra of dimension
$n$ over a finite field $k$ and let $\Aalg$ be given in the input in terms of basis 
elements over $k$. Then there is a deterministic
algorithm which in subexponential time $poly(n^{\log n}, \log |k|)$ computes
a decomposition of $\Aalg$ into a direct sum 
$\Aalg_1\oplus\ldots\oplus \Aalg_t$ and finds an automorphism of order 
$\dim_k\Aalg_i$ of the algebra $\Aalg_i$, for each $1\le i\le t$.}
\\

This main theorem can be considered as a GRH-free version of Evdokimov's factoring result
\cite{Ev}, but its proof leads us to significantly generalize standard notions and 
develop novel algebraic techniques that suggest a general paradigm for GRH elimination.
We are going to use it
as a tool for more important applications but first let 
us explain the importance of this result itself.
It is the first deterministic subexponential time algorithm to find a nontrivial 
automorphism of a given commutative semisimple algebra over a finite field. Finding
a nontrivial automorphism of a given arbitrary ring is in general as hard as integer
factoring \cite{ks06} but our result shows that it might be a lot easier for 
a commutative semisimple algebra over a finite field.
Note that in the special case when $\Aalg=k[x]/(f(x))$ with $f(x)$ splitting over 
$k$ as $\prod_{j=1}^n$ $(x-\alpha_j)$, with $\alpha_1,\ldots,\alpha_n$ 
all distinct,
we have $\Aalg\cong\oplus_{j=1}^n$ 
$k[x]/(x-\alpha_j)$. The above algorithm either gives $t>1$ components of $\Aalg$ 
-- in which case it effectively yields a nontrivial factor of $f(x)$ -- or $t=1$ and 
it gives an automorphism $\sigma$ of $\Aalg$ of order $n$, thus yielding $n$ 
distinct ``roots'' of $f(x)$ -- $x$, $\sigma(x),\ldots$, $\sigma^{n-1}(x)$ -- all 
living in $\Aalg\setminus k$. This latter case can be interpreted as finding roots 
over finite fields in terms of ``radicals'', in analogy to classical Galois theory where
one studies rational polynomials whose roots can be expressed by radicals,
see Section~\ref{sec-kummer-extn} for details. 

The key ideas in finding a nontrivial automorphism of a given commutative 
semisimple $\Balg$-algebra $\Aalg$ over a finite field $k\subseteq\Balg$ are as follows. We 
consider a special ideal $\Aalg^\prime$ (what we call the {\em essential part} 
\iffull in Section \ref{sec-essential}\fi 
) of the tensor product $\Aalg\otimes_\Balg\Aalg$. The ideal $\Aalg'$ is just
the kernel of a standard homomorphism of $\Aalg\otimes_\Balg\Aalg$ onto
$\Aalg$ and has rank (``dimension'') 
$\rk_\Balg\Aalg(\rk_\Balg\Aalg-1)$ over $\Balg$. 
The algebra $\Aalg$ gets naturally 
embedded in $\Aalg^\prime$ by a map $\phi$, hence $\Aalg^\prime$ is an extension 
algebra of $\phi(\Aalg)\cong\Aalg$ which in turn is an extension algebra of 
$\phi(\Balg)\cong\Balg$. Also, we know a natural automorphism of $\Aalg^\prime$ fixing 
$\Balg$ -- the map $\tau:x\otimes y\mapsto$ $y\otimes x$. 
A lot of technical effort 
goes into ``bringing down'' this automorphism 
(or certain other automorphism $\sigma$ of order $2$
obtained by recursion)
from $\Aalg^\prime$ to $\Aalg$, i.e. getting a
$\Balg$-automorphism $\sigma^\prime$ of $\Aalg$. The technical arguments fall into two 
cases, depending on whether $\rk_\Aalg\Aalg'=\rk_\Balg\Aalg'/\rk_\Balg\Aalg$ is odd or 
even.

{\bf (1)} If the rank $\rk_\Balg\Aalg$ is even then
$\rk_\Aalg\Aalg'$ is odd. We find an element 
$u\in\Aalg^\prime$ with $u^\tau=-u$. If $u\in\Aalg$
then the restriction of $\tau$ is a $\Balg$-automorphism
of the subalgebra $\Balg[u]$ of $\Aalg$ 
generated by $\Balg$ and $u$. If 
$u\not\in\Aalg$ then either the subalgebra 
$\Aalg[u]$ of $\Aalg'$ is not a free
$\Aalg$-module or $\Aalg'$ is not a free
$\Aalg[u]$-module. Both cases give us a zero divisor 
in $\Aalg^\prime$ to go to a smaller 
ideal $\Iid$ of $\Aalg^\prime$ 
such that we know an automorphism of $\Iid$, 
it contains a ``copy'' of $\Aalg$ and 
$\rk_\Aalg\Iid$ is odd, thus we can continue this 
``descent'' (from $\Aalg'$ to $\Iid$) till we have a $\Balg$-automorphism 
of $\Aalg$ or of a subalgebra of $\Aalg$ 
\iffull (this process appears in Section \ref{sec-bring-down})\fi.
In the former case we are done while in
the latter case we use two recursive calls 
and certain techniques to ``glue'' the 
three available automorphisms.
{\bf (2)}
If the rank $\rk_\Balg\Aalg$ is odd then
$\rk_\Aalg\Aalg'$ is even
and we can use the technique above to 
find an $\Aalg$-automorphism $\sigma$
of $\Aalg'$. It turns out that
$\sigma$ and $\tau$ generate a group
of automorphisms of $\Aalg'$ which is big enough 
to find a proper ideal $\Iid$ of $\Aalg'$ efficiently.
We may further assume that the rank of
$\Iid$ over $\Aalg$ is at most 
$\rk_\Aalg\Aalg'/2=(\rk_\Balg\Aalg-1)/2$. This allows
us a recursive call with $(\Iid,\Aalg)$ in place
of $(\Aalg,\Balg)$ to get an $\Aalg$-automorphism of $\Iid$,
which we eventually show is enough to extract an automorphism of 
$\Aalg$ using tensor properties and a recursive call 
\iffull(this case 2 gets handled in \ref{Evd-sect})\fi. 

This algebraic-extensions jugglery {\em either} goes through and 
yields a nontrivial automorphism $\sigma^\prime$ of $\Aalg$ 
fixing $\Balg$ {\em or} it ``fails'' and yields a zero divisor 
in $\Aalg$ which we use to ``break'' $\Aalg$ into smaller subalgebras 
and continue working there. 
As in each recursive call, in the above two cases, the rank of the bigger algebra 
over the subalgebra is at most half of the original
one, the depth of the recursion is at most $\log\rk_\Balg\Aalg$.
This gives an $n^{\log n}$ term in the time complexity analysis.


Roots of unity play a significant role in gluing
automorphisms (i.e. in extending an automorphism of
a subalgebra, of elements fixed by another
automorphism, to the whole algebra). 
\iffull The gluing process is described in Section \ref{sec-glue}. \fi 
As we do not know roots of unity in $k$ we resort to attaching 
virtual $r$-th roots of unity for a suitable prime $r$, i.e. working in the cyclotomic extension 
$k[\zeta_r]:=k[x]/(\sum_{i=1}^{r-1}x^i)$ and $\Aalg^\prime[\zeta_r]:=$ 
$k[\zeta_r]\otimes_k\Aalg^\prime$. We then need to generalize standard algebraic 
constructions, like {\em Kummer extensions} and {\em Teichm\"uller subgroups} which were 
first used in a context similar to ours
by Lenstra \cite{Len} to find isomorphisms between fields, to our situation of 
commutative semisimple algebras. 

\medskip

The above theorem and its proof techniques have important applications. The first 
one is in finding zero divisors in a noncommutative algebra.

\smallskip
\noindent
{\bf Application 1: }
{\em Let $\Aalg$ be an algebra of dimension $n$ over a finite field $k$ and let 
$\Aalg$ be given in the input in terms of basis elements over $k$. Assume that 
$\Aalg$ is noncommutative. Then there is a deterministic algorithm which finds a 
zero divisor in $\Aalg$ in time $poly(n^{\log n}, \log |k|)$. }

The previous best result was due to R\'onyai \cite{ro3} who gave an 
algorithm invoking polynomial factorization over finite fields and hence taking 
subexponential time assuming GRH. Our result removes the GRH assumption. It is
interesting to note that if we prove such a result for {\em commutative} algebras
as well then we would basically be able to factor polynomials in subexponential 
time without needing GRH. 

If $\Aalg$ is a simple algebra over the finite field $k$ then
it is isomorphic to the algebra $M_m(K)$ of the $m\times m$ matrices with
entries from an extension field $K$ of $k$. By Application~1
we find a proper left ideal of $\Aalg$. A recursive call
to a certain subalgebra of the left ideal will ultimately give
a minimal left ideal of $\Aalg$ and using this minimal one-sided 
ideal an isomorphism with $M_m(K)$ can be efficiently computed.
Thus, for constant $m$, Application~1 extends Lenstra's result (on computing 
isomorphisms between input fields) to
noncommutative simple algebras, i.e, the {\em explicit isomorphism problem}
is solved in this case. We note that, in general, algebra isomorphism problem 
over finite fields is not ``believed'' to be NP-hard but it is at least as hard as 
the graph isomorphism problem \cite{ks06}. We also remark that the analogous
problem of constructing isomorphism with the algebra of matrices
over the rationals has a surprising application to rational
parametrization of certain curves, see \cite{ghps}.

\bigskip

The techniques used to prove Main Theorem can be applied to find a nontrivial 
factor of an $r$-th cyclotomic polynomial over a finite field $k$, for almost all
$r$'s, in deterministic polynomial time. 

\smallskip
\noindent
{\bf Application 2: }
{\em Let $r$ be a positive integer such that the multiplicative group $\Z_r^*$ is noncyclic and let 
$\phi_r(x)$ be the $r$-th 
cyclotomic polynomial. Then we can find a nontrivial factor of $\phi_r(x)$ over a finite field $k$ in 
deterministic $poly(r,\log |k|)$ time.}

Roots of an $r$-th cyclotomic polynomial over $k$ are the $r$-th roots of unity and
thus naturally related to all polynomial factoring algorithms. Assuming GRH several
algorithms are known to factor these important polynomials (see \cite{ev-solvable}).
The above result gives the first deterministic polynomial time algorithm to 
nontrivially factor ``most'' of the cyclotomic polynomials without assuming GRH.

\bigskip

The third application of the techniques used to prove Main Theorem is in the instance
of polynomial factoring over prime fields when we know the Galois group of the input 
polynomial. The following theorem can be seen as the GRH-free version of the main 
theorem of R\'onyai \cite{ro4}.

\smallskip
\noindent
{\bf Application 3: }
{\em Let $F(X)\in\Z[X]$ be a polynomial irreducible over $\Q$ with Galois group of size 
$m$ and let $L$ be the maximum length of the coefficients of $F(X)$. Let $p$ be a 
prime not dividing the discriminant of $F(X)$ and let $f(x)=F(X)\pmod{p}$. Then by a 
deterministic algorithm of running time $poly(m,L,\log p)$ we can find either a 
nontrivial factor of $f(x)$ or a nontrivial automorphism of $\F_p[x]/(f(x))$ of order 
$\deg f$.}

\bigskip

The fourth application of the techniques used to prove Main Theorem is in the instance
of polynomial factoring over $\F_p$ when $p$ is a prime with smooth $(p-1)$. The 
following theorem can be seen as the GRH-free version of the main theorem of R\'onyai 
\cite{ro2}.

\smallskip
\noindent
{\bf Application 4: }
{\em Let $f(x)$ be a polynomial of degree $n$, that splits into linear factors 
over $\F_p$. Let $r_1<\ldots<r_t$ be the prime factors of $(p-1)$. Then by a deterministic 
algorithm of running time $poly(r_t,n,\log p)$, we can find either a nontrivial factor of 
$f(x)$ or a nontrivial automorphism of $\F_p[x]/(f(x))$ of order $n$. In fact, we always find
a nontrivial factor of $f(x)$ in case $n\not|\text{ lcm}\{r_i-1 | 1\le i\le t\}$.}

Thus over ``special'' fields (i.e. when $p-1$ has only small prime factors) the above 
actually gives a deterministic polynomial time algorithm, a significant improvement over 
Main Theorem.  

\iffull
\subsection{Organization}

In Section \ref{sec-prelim} we collect various standard objects and structural
facts associated to algebras. We also discuss the three basic methods that lead
to discovering a zero divisor in an algebra -- finding discrete log for elements 
of prime-power order, finding a free base of a module and refining an ideal by 
a given automorphism.

In this work we use methods for finding zero divisors in algebras in the case 
when certain groups of automorphisms are given. One of such methods is computing 
fixed subalgebras and testing freeness over them. In Section~\ref{sec-semireg} we give 
a characterization of algebras and groups which survive these kinds of attacks. 
These algebras, called {\em semiregular} wrt the group, behave like fields in the
sense that the whole algebra is a free module over the subalgebra of fixed points 
of the group and the rank equals the size of the group.

In Section~\ref{sec-kummer-extn} we build a small theory for the main algebraic 
construction, {\em Kummer-type extensions} over algebras, that we are going to 
use. We investigate there the action of the automorphisms of an algebra $\Aalg$ 
on a certain 
subgroup, {\em Teichm\"uller subgroup}, of the multiplicative group of a Kummer-type 
extension of $\Aalg$. The proofs of Applications 2 and 3 get completed in this section. 

In Section~\ref{sec-findaut} we apply the machinery of Section~\ref{sec-kummer-extn}
to the tensor power algebras and complete the proof of Main Theorem.

In Section \ref{sec-noncomm} we find suitable subalgebras of a given noncommutative 
algebra to invoke Main Theorem and complete the proof of Application 1.

In Section~\ref{sec-p-1} we use the techniques developed for the Main Theorem in the
case of special finite fields and complete the proof of Application 4.

\fi

\section{Preliminaries}\label{sec-prelim}

\iffull

In this section we list some algebraic notions that we use in this work and that can be 
found in standard algebra texts, for example \cite{Lang}.



\smallskip
\noindent
{\bf Rings, Units and Zero-divisors:}
A {\em ring with identity} (or ring, for short) $R$ is a set of elements together with two 
operations -- denoted
by {\em addition} $+$ and {\em multiplication} $\cdot$ -- such that $(R,+)$ is an
Abelian group, $\cdot$ is associative, distributes over $+$ and has an {\em identity}
element $1_R$. Note that the set $R^*$, containing all the elements of $R$ that 
have a multiplicative inverse, is a multiplicative group called the group of {\em units}. For a 
prime integer $r$ we call a unit $x$ an {\em $r$-element} if the multiplicative order of $x$
is a power of $r$. An element $x$ is called a {\em zero divisor} if $x\not=0$ and there exist 
nonzero $y,y^\prime\in\Aalg$ such that $yx=xy^\prime=0$. 

\smallskip
\noindent
{\bf Modules:}
Let $(R,+,\cdot)$ be a commutative ring and $(M,+)$ be an Abelian group. We call $M$ 
an {\em $R$-module} wrt an operation $R\times M\rightarrow M$ (called {\em scalar multiplication} 
and denoted as $rx$ for $r\in R$ and $x\in M$) if for all $r,s\in R; x,y\in M$, we have:
$r(x+y) = rx+ry$; $(r+s)x = rx+sx$; $(rs)x = r(sx)$ and $1x = x$.
Note that a vector space $V$ over a field $\F$ is also an $\F$-module.

\smallskip
\noindent
{\bf Free and Cyclic:}
For an $R$-module $M$, a set $E\subset M$ is called a {\em free basis} of $M$ 
if: $E$ is a {\em generating set} for $M$, i.e. every element of $M$ is a finite sum of elements of 
$E$ multiplied by coefficients in $R$, and $E$ is a {\em free set}, i.e. for all $r_1,\ldots,r_n\in R$; 
$e_1,\ldots,e_n\in E$,  $r_1e_1+\cdots+ r_ne_n=0$ implies that $r_1=\cdots= r_n = 0$.
A {\em free module} is a module with a free basis. $|E|$ is called the {\em rank} or {\em 
dimension} of the free module $M$ over $R$. Clearly, a vector space is a free module.
A module is called a {\em cyclic} module if it is generated by one element.

\smallskip
\noindent
{\bf Algebras:}
Let $(R,+,\cdot)$ be a commutative ring and $(\Aalg,+,\cdot)$ be a ring 
which is also an $R$-module, where the additive
operation of $\Aalg$ as a module coincides with $+$. We say that 
$\Aalg$ is an associative $R$-algebra with identity
(or just an {\em $R$-algebra} for short) if multiplication 
by elements of $R$ commutes with multiplication by elements of
$\Aalg$: for every $r\in R$ and for every $a,b\in\Aalg$ we have
$r(ab)=(ra)b=a(rb)$.

\smallskip
\noindent
{\bf Subalgebras:}
A {\em subalgebra} $\Balg$ of an $R$-algebra $(\Aalg,+,\cdot)$ is just a
submodule of $\Aalg$ closed under multiplication. In this paper unless
otherwise stated, by a subalgebra of $\Aalg$ we mean a subalgebra 
containing the identity element $1_\Aalg$. Note that 
if $\Balg$ is a commutative subalgebra of $\Aalg$ then
$\Aalg$ is a $\Balg$-module in a natural way. 
If, furthermore, $\Balg$ is contained in the center of 
$\Aalg$ (that is, $ab=ba$ for every $a\in\Aalg$ and
for every $b\in\Balg$) then $\Aalg$ is a $\Balg$-algebra.


\smallskip
\noindent
{\bf Presentation:}
In this work we will consider only $k$-algebras $\Aalg$ 
that are {\em finite dimensional} over a finite field $k$. 
So we can assume that an algebra $\Aalg$ is always presented 
in the input-output in terms of an additive basis of $(\Aalg,+)$ over $k$, i.e. there are {\em 
basis elements} $b_1,\ldots,b_n\in\Aalg$ such that $\Aalg=kb_1+\cdots+kb_n$ and furthermore 
$a_{i,j,\ell}\in k$ are given such that $b_i\cdot b_j=$ $\sum_{\ell}a_{i,j,\ell}b_\ell$.
Such an $n$ is called the {\em dimension}, $dim_k\Aalg$, of $\Aalg$ over $k$. 

\smallskip
\noindent
{\bf Extension:}
If $\Balg$ is a commutative $k$-algebra and a $\Balg$-algebra 
$\Aalg$ is also a free module over $\Balg$ then we call $\Aalg$ an 
{\em algebra extension} or an {\em extension algebra} over $\Balg$. This
terminology is justified by the fact that $\Balg$ is embedded into
(the center of) $\Aalg$ by the map $b\mapsto b1_\Aalg$. We denote the
rank (``dimension") of $\Aalg$ as a $\Balg$-module by 
$\rk_\Balg\Aalg$ or $[\Aalg:\Balg]$. We sometimes use this 
notation also when there is an implicit embedding of $\Balg$ in $\Aalg$.

\smallskip
\noindent
{\bf Primitive Element:}
We call an algebra extension $\Aalg$ over $\Balg$ {\em simple} if there is an $\alpha\in\Aalg$
such that $\{1,\alpha,\ldots$, $\alpha^{n-1}\}$ forms a free basis of $\Aalg$ over $\Balg$. We 
call $\alpha$ a {\em primitive element} and write $\Aalg=\Balg[\alpha]$.

Following is a version of the standard {\em Primitive Element Theorem}.
\begin{fact}\label{fac-primitive}
If $K\supseteq F$ are fields such that $char\ F$ is $0$ or $>[K:F]^2$, then $K$ has a primitive
element over $F$.
\end{fact}

There are two natural operations defined on algebras -- the {\em direct sum} and the
{\em tensor product} -- each constructs a bigger algebra.

\smallskip
\noindent
{\bf Direct Sum:}
Let $(\Aalg_1,+,\cdot)$ and $(\Aalg_2,+,\cdot) $ be two algebras. Then the {\em direct sum 
algebra}, $\Aalg_1\oplus\Aalg_2$, is the set $\{(a_1,a_2)\mid a_1\in\Aalg_1, a_2\in\Aalg_2\}$ 
together with component-wise addition and multiplication operations. In a similar vein, for 
subalgebras $\Aalg_1, \Aalg_2$ of an algebra $\Aalg$ we write $\Aalg=\Aalg_1\oplus\Aalg_2$, if
$\Aalg=\Aalg_1+\Aalg_2$ and $\Aalg_1, \Aalg_2$ are {\em orthogonal} i.e. $\forall$ 
$a_1\in\Aalg_1$, $a_2\in\Aalg_2$, $a_1a_2=a_2a_1=0$. 

\smallskip
\noindent
{\bf Tensor Product:}
Furthermore, if $\Balg$ is a commutative algebra such that $\Aalg_1, \Aalg_2$ are 
$\Balg$-algebras of dimensions $n_1, n_2$ respectively over $\Balg$ then their {\em tensor 
product algebra wrt } $\Balg$, $\Aalg_1\otimes_\Balg\Aalg_2$, is the set 
$\{a_1\otimes a_2\mid a_1\in\Aalg_1, a_2\in\Aalg_2\}$ naturally viewed as a $\Balg$-module 
having the multiplication operation: $(a_1\otimes a_2)\cdot$ 
$(a_1^\prime\otimes a_2^\prime)=$ $(a_1a_1^\prime \otimes a_2a_2^\prime)$ for all 
$a_1, a_1^\prime\in\Aalg_1$ and $a_2, a_2^\prime\in\Aalg_2$. Note that the tensor product 
algebra has dimension $n_1n_2$ over $\Balg$. Thus, if $\Balg$ is finite then 
$|\Aalg_1\oplus\Aalg_2|=|\Balg|^{n_1+n_2}$, while 
$|\Aalg_1\otimes_\Balg\Aalg_2|=|\Balg|^{n_1n_2}$.

\smallskip
\noindent
{\bf Nilpotent and Idempotent:}
In an algebra $\Aalg$ we call an element $x\in\Aalg$ {\em nilpotent} if $x^m=0$ for some
$m\in\Z$, while we call $x$ {\em idempotent} if $x^2=x$. It is called a {\em primitive} 
idempotent if it cannot be expressed as the sum of two idempotents whose product is zero.
It is called {\em nontrivial} if it is not $0$ or $1$.

\smallskip
\noindent
{\bf Decomposability:}
An algebra $\Aalg$ is called {\em indecomposable} if there are no nonzero algebras $R, S$
such that $\Aalg\cong R\oplus S$.

Following are some standard facts relating decomposability to idempotents in commutative 
algebras.

\begin{fact}\label{fac-idempotent}
Let $\Aalg$ be a commutative algebra then:
\\
(1) $\Aalg$ decomposes iff $\Aalg$ has a nontrivial idempotent.  
\\
(2) If $e$ is an idempotent in $\Aalg$ then $\Aalg\cong$ 
$e\Aalg\oplus (1-e)\Aalg$.
\\
(3) If $e$ is a primitive idempotent in $\Aalg$ then $e\Aalg$ is indecomposable. 
\end{fact}

\smallskip
\noindent
{\bf Ideal:}
An {\em ideal} $I$ of an algebra $\Aalg$ is a subset that is an additive subgroup of 
$\Aalg$, is closed under multiplication and it contains both $aI:=\{a\cdot i \mid i\in I\}$;
$Ia:=\{i\cdot a \mid i\in I\}$ for all $a\in\Aalg$. Note that $\{0\}$ and $\Aalg$ are 
ideals of $\Aalg$, we call them {\em trivial} ideals. Also note that
proper ideals are not subalgebras in the strict sense used in this paper.

\smallskip
\noindent
{\bf Semisimplicity:}
An algebra $\Aalg$ is called {\em simple} if it has no nontrivial ideal. An algebra
is called {\em semisimple} if it is a direct sum of simple algebras.

Following are some standard facts about commutative semisimple algebras.

\begin{fact}\label{fac-comm-semisimple}
Let $\Aalg$ be a commutative semisimple algebra then: 
\\
(1) $\Aalg$ is a direct sum of fields.
\\
(2)  If $I$ is an ideal of $\Aalg$ and $I^\perp:=\{a\in\Aalg\mid aI=0\}$ (called the {\em 
complement} of $I$) then $\Aalg=I\oplus I^\perp$. 
Furthermore, there exists an idempotent $e$ of $\Aalg$ such that $I=e\Aalg$ thus giving an
explicit projection from $\Aalg$ to $I$. 
\end{fact}

Following is the celebrated {\em Artin-Wedderburn Theorem} that classifies semisimple 
algebras.

\begin{fact}\label{fac-wedderburn}
Any semisimple algebra $\Aalg$ is isomorphic to a direct sum of $n_i\times n_i$ matrix 
algebras over division rings $D_i$ (i.e. $D_i$ satisfies all field axioms except 
commutative multiplication). Both the $n_i$'s and $D_i$'s are uniquely determined up to 
permutation of the indices $i$.
\end{fact}

\smallskip
\noindent
{\bf Morphisms:}
Let $\phi$ be a map between two algebras $\Aalg$, $\Balg$. If $\phi$ preserves the addition 
and multiplication operations of the algebras then we call it a {\em homomorphism}. If the 
homomorphism $\phi$ is injective then we call it an {\em embedding}. If the homomorphism 
$\phi$ is both injective and surjective then we call it an {\em isomorphism}. A 
homomorphism from an algebra to itself is called an {\em endomorphism}. An isomorphism
from an algebra to itself is called an {\em automorphism}. A set $S$ is said to be 
{\em invariant} under the automorphism $\phi$ of $\Aalg$ if for all $s\in S$, $\phi(s)\in S$.
$\phi$ is said to {\em fix} $S$ if $\phi$ fixes each element of $S$, i.e. for all $s\in S$, 
$\phi(s)=s$. The group of $S$-automorphisms of $\Aalg$, $Aut_S(\Aalg)$, is the set of all 
automorphisms of $\Aalg$ that fix $S$.

\fi

Throughout this paper all algebras are algebras with identity elements. 
Unless otherwise stated explicitly, by a subalgebra we mean a subalgebra 
{\it containing the identity element}. Thus, in this strict sense 
a proper ideal is {\it not} considered as a subalgebra. 
In the rest of this section $\Aalg$ stands for a commutative semisimple
algebra over the finite field $k$.

\iffull
\subsection{Discrete Log for $r$-elements}
\fi

Given two $r$-elements (i.e. having order a power of the prime $r$) in a commutative semisimple 
algebra there is an algorithm that computes the discrete logarithm or finds a zero divisor (of a 
special form) in $\Aalg$. 
\iffull
We describe this algorithm below, it 
\fi
is a variant of the Pohlig-Hellman 
\cite{pohlig} algorithm with the equality testing of elements replaced by testing whether their 
difference is a zero divisor.

\begin{lemma}\label{lem-discretelog}
\iffull
Given a prime $r$ distinct from the characteristic of a finite field $k$, a finite dimensional 
commutative semisimple algebra $\Aalg$ over $k$ and two 
$r$-elements $a,b\in\Aalg^*$, 
\fi
such 
that the order of $a$ is greater than or equal to the order of $b$. There is a deterministic 
algorithm which computes in time $poly(r, \log |\Aalg|)$:
\iffull
\\(1) 
\fi
either two non-negative integers $s,s^\prime$ such that $a^s-b^{s^\prime}$ is a  
zero divisor in $\Aalg$,
\iffull
\\(2) 
\fi
or an integer $s\geq 0$ with $a^s=b$.
\end{lemma}
\iffull
\begin{proof}
Let $t_a$ be the smallest non negative integer such that $a^{r^{t_a}}-1$ is zero or a  
zero divisor in $\Aalg$. Since $t_a\leq \log_r|\Aalg|$ we can compute 
$a^{r^0}-1, a^{r^1}-1, \ldots, a^{r^{t_a}}-1$ in
$poly(\log|\Aalg|)$ time via fast exponentiation.
We are done if $0\not=a^{r^{t_a}}-1=a^{r^{t_a}}-b^0$ is a
zero divisor. Therefore we may assume that $a^{r^{t_a}}=1$, i.e. 
the order of $a$ is $r^{t_a}$. Let $t_b$ be the smallest non-negative integer 
such that $b^{r^{t_b}}-1$ is a zero divisor. Like $t_a$,
$t_b$ can be computed in polynomial time and we may again assume that
$r^{t_b}$ is the order of $b$. Replacing $a$ with $a^{r^{t_a-t_b}}$ we may 
assure that $t_a=t_b=t$. In this case for every primitive idempotent $e$ of $\Aalg$: $ea, eb$ 
have order $r^t$ in the finite field $e\Aalg$. As the multiplicative group of a finite field is cyclic, 
this means that there exists a nonnegative integer $s<r^t$ such that $(ea)^s=eb$. So
we now attempt to find this discrete log, $s$, and the corresponding idempotent $e$ as well.

We iteratively compute the consecutive sections of the base $r$ expansion of $s$. 
To be more specific, we compute integers $s_0=0,s_1,s_2,\ldots,s_t$ together with 
idempotents $e_1,\ldots,e_t$ of $\Aalg$ such that, for all $1\le j\le t$: $0\leq s_j<r^j$, 
$s_j\equiv s_{j-1}\pmod{r^{j-1}}$ and $a^{{s_j}r^{t-j}}e_j=b^{r^{t-j}}e_j$.

In the initial case $j=1$ we find by exhaustive search, in at most $r$ rounds, an
$s_1\in\{1,\ldots,r-1\}$ such that $z_1=$ $(a^{r^{t-1}{s_1}}-b^{r^{t-1}})$
is zero or a zero divisor. If it is zero then we set $e_1=1$ otherwise we compute and 
set $e_1$ equal to the identity element of the annihilator ideal $\{x\in\Aalg|z_1x=0\}$.

Assume that for some $j<t$ we have found already $s_j$ and $e_j$ with the desired property.
Then we find by exhaustive search, in at most $r$ rounds, an integer 
$d_{j+1}\in \{0,\ldots,r-1\}$ such that $z_{j+1}=$ 
$(a^{(s_j+r^jd_{j+1})r^{t-j-1}}-b^{r^{t-j-1}})$ is zero or a zero divisor. We set 
$s_{j+1}=(s_j+d_{j+1}r^j)$ and take as $e_{j+1}$ the identity element of the annihilator 
ideal $\{x\in e_j\Aalg|xz_{j+1}=0\}$. 

The above procedure clearly terminates in $t$ rounds and using fast exponentiation
can be implemented in $poly(r, \log |\Aalg|)$ time.
\end{proof}
\fi

\iffull
\subsection{Free Bases of Modules}
\fi

One of the possible methods for finding zero divisors in algebras is attempting to compute 
a free basis of a module over it. 
\iffull
Following Lemma states the basic tool to do that.
\fi

\begin{lemma}\label{lem-nonfree}
Let $V$ be a finitely generated module 
\iffull
over a finite dimensional algebra $\Aalg$ over a finite field $k$. 
\fi
If $V$ is not a free $\Aalg$-module
then one can find a zero divisor in $\Aalg$ deterministically in time 
$poly(\dim_\Aalg V, \log |\Aalg|)$.
\end{lemma}
\iffull
\begin{proof}
We give an algorithm that attempts to find a free basis of $V$ over $\Aalg$, but as there is
no free basis it ends up finding a zero divisor.

Pick a nonzero $v_1\in V$. We can efficiently check whether a nonzero $x\in\Aalg$ exists 
such that $xv_1=0$, and also find it by linear algebra over $k$. If we get such an $x$ then it 
is a zero divisor, for otherwise $x^{-1}$ would exist implying $v_1=0$. So suppose such an 
$x$ does not exist, hence $V_1:=\Aalg v_1$ is a free $\Aalg$-module. Now 
$V_1\neq V$ so find a $v_2\in V\setminus V_1$ by linear algebra over $k$. Again we can 
efficiently check whether a nonzero $x\in\Aalg$ exists such that $xv_2\in V_1$, and also find 
it by linear algebra over $k$. If we get such an $x$ then it is a zero divisor, for otherwise 
$x^{-1}$ would exist implying $v_2\in V_1$.  So suppose such an $x$ does not exist, hence 
$V_2:=\Aalg v_1+\Aalg v_2$ is a free $\Aalg$-module. Now $V_2\neq V$ so we can find a 
$v_3\in V\setminus V_2$ by linear algebra over $k$ and continue this process. This process
will, in at most $dim_\Aalg V$ iterations, yield a zero divisor as $V$ is not a free 
$\Aalg$-module. 
\end{proof}
\fi

\iffull
\subsection{Automorphisms and Invariant Ideal Decompositions}\label{autgrp-sect}
\fi


Automorphisms of $\Aalg$ are assumed to be given as linear transformations 
of the $k$-vector space $\Aalg$ in terms of a $k$-linear basis of $\Aalg$. 
For images we use the superscript notation while for the fixed 
points the subscript notation: if $\sigma$ is an automorphism of $\Aalg$ 
then the image of $x\in\Aalg$ under $\sigma$ is denoted by $x^\sigma$. 
If $\Gamma$ is a set of automorphisms of $\Aalg$ then $\Aalg_\Gamma$ 
denotes the set of the elements of $\Aalg$ fixed by every 
$\sigma\in \Gamma$. It is obvious that $\Aalg_\Gamma$ is a subalgebra of 
$\Aalg$. For a single automorphism $\sigma$ we use $\Aalg_\sigma$ in 
place of $\Aalg_{\{\sigma\}}$.

Given an ideal $I$ of $\Aalg$ and an automorphism $\sigma$ of $\Aalg$ we 
usually try to find zero divisors from the action of $\sigma$ on $I$. 
\iffull
Note that, by Fact \ref{fac-comm-semisimple}, $\Aalg=I\oplus I^\perp$. 
Now 
$I^\sigma$ is an ideal of $\Aalg$, and if it is neither $I$ nor $I^\perp$
then we try computing $I\cap I^\sigma$. This can be easily computed by 
first finding the identity element $e$ of $I$, and then $I\cap I^\sigma$ 
is simply $\Aalg ee^\sigma$. By the hypothesis this will be a proper ideal
of $I$, thus leading to a {\em refinement} of the decomposition: 
$\Aalg=I\oplus I^\perp$. This basic idea can be carried all the way to give
the following tool that finds a refined, invariant, ideal decomposition.
\fi

\begin{lemma}\label{lem-refine}
\iffull
Given $\Aalg$, a commutative semisimple algebra over a finite field $k$
together with a set of $k$-automorphisms $\Gamma$ of $\Aalg$
and a decomposition of $\Aalg$ into a sum of pairwise orthogonal ideals 
$J_1,\ldots,J_s$, there is a deterministic algorithm of time complexity
$poly(|\Gamma|, \log |\Aalg|)$ that computes a decomposition of $\Aalg$ 
into a sum of pairwise orthogonal ideals $I_1,\ldots,I_t$ such that: 
\\(1) the new decomposition is a refinement of the original one --
for every $j\in\{1,\ldots,t\}$, there exists $i\in\{1,\ldots,s\}$ such that
$I_j\subseteq J_i$, and
\\(2) the new decomposition is invariant under $\Gamma$ -- the group 
generated by $\Gamma$ permutes the ideals $I_1,\ldots,I_t$,
i.e. for every $\sigma\in \Gamma$ and for every index $j\in \{1,\ldots,t\}$, 
we have $I_j^\sigma=I_{j^\sigma}$ for some index $j^\sigma\in \{1,\ldots,t\}$.
\fi
\end{lemma}

\section{Semiregularity}\label{sec-semireg}

In this section we 
continue to assume that $\Aalg$ is 
a commutative semisimple algebra over a finite field $k$.
Given $\Gamma\subseteq\Aut_k(\Aalg)$, a basis of $\Aalg_\Gamma$ 
can be computed by solving a system of linear equations in $\Aalg$.
Thus, we can apply the method of Lemma \ref{lem-nonfree} considering $\Aalg$ as a 
$\Aalg_\Gamma$-module wrt the multiplication in 
$\Aalg$. 
\iffull
In this section we describe 
\fi
a class of algebras, together with automorphisms, 
that are free modules over the subalgebra of the fixed points of the corresponding set 
of automorphisms, i.e. on which the tool of Lemma \ref{lem-nonfree} is ineffective. 

Let $\sigma$ be a $k$-automorphism of $\Aalg$. We 
say that $\sigma$ is {\em fix-free} if there is no nontrivial ideal $I$ of $\Aalg$ such 
that $\sigma$ fixes $I$. We call a group $G\leq \Aut(\Aalg)$ {\em semiregular} if every 
non-identity element of $G$ is fix-free. A single automorphism $\sigma$ of $\Aalg$ is 
{\em semiregular} if $\sigma$ generates a semiregular group of automorphisms of $\Aalg$. 

We have the following characterization of semiregularity.

\begin{lemma}\label{lem-semireg-ineq}
\iffull
Let $\Aalg$ be a commutative semisimple algebra over a finite field $k$ and 
let 
\fi
$G$ 
be a group of $k$-automorphisms of $\Aalg$. Then 
$\dim_k\Aalg\leq |G|\cdot \dim_k\Aalg_G$, where equality holds if and only if $G$ is 
semiregular. 
This condition is also equivalent to saying that 
$\Aalg$ is a free $\Aalg_G$-module of rank $|G|$.
\end{lemma}
\iffull
\begin{proof}
The proof is based on the observation that $\Aalg$ is a direct sum of fields and a 
$k$-automorphism of $\Aalg$ just {\em permutes} these component fields.  

Let $e$ be a primitive idempotent of $\Aalg$. We denote
the stabilizer of $e$ in $G$ by $G_e$, i.e, $G_e=\{\sigma\in G|e^\sigma=e\}$.
Let $C$ be a complete set of right coset representatives modulo
$G_e$ in $G$. The orbit of $e$ under $G$ is $\{e^\gamma|\gamma\in C\}$ and they are 
$|G:G_e|$ many pairwise orthogonal primitive idempotents in $\Aalg$. This means that 
the component field $e\Aalg$ is sent to the other component fields 
$\{e^\gamma\Aalg|\gamma\in C\}$ by $G$. Thus, the element $f:=\sum_{\gamma\in C}e^\gamma$ 
$\in\Aalg_G$ is a primitive idempotent of $\Aalg_G$ and equivalently $f\Aalg_G$ is a 
field.

The subgroup $G_e$ acts as a group of field automorphisms of $e\Aalg$. This gives a 
restriction map $\lambda:G_e\rightarrow Aut_k(e\Aalg)$ whose kernel say is $N_e$, so
$N_e=\{\sigma\in G|\sigma \text{ fixes } e\Aalg\}$ is a normal subgroup of $G_e$, thus $G_e/N_e$ are
distinct $k$-automorphisms of the field $e\Aalg$.
We claim that $(e\Aalg)_{G_e}=e\Aalg_G$. The inclusion $e\Aalg_G\subseteq (e\Aalg)_{G_e}$ 
is trivial. To see the reverse inclusion, let $x\in (e\Aalg)_{G_e}$ and consider 
$y:=\sum_{\gamma\in C}x^\gamma$. Since $x\in e\Aalg$ we get $ex=x$ and 
$y=\sum_{\gamma\in C}e^\gamma x^\gamma$, whence using the orthogonality of the idempotents 
$e^\gamma$, we infer $ey=x$. The fact that $y\in\Aalg_G$ completes the proof of the claim. As $G_e$
is a group of automorphisms of the field $e\Aalg$, this claim implies $e\Aalg_G$ is a field too and also by
Galois theory $[e\Aalg:e\Aalg_G]=|G_e/N_e|$. 

Observe that $ef=e$ and this makes multiplication by $e$ a onto homomorphism from $f\Aalg_G$ to 
$e\Aalg_G$. This homomorphism is also injective as $e\Aalg_G$, $f\Aalg_G$ are fields, thus making
$f\Aalg_G\cong e\Aalg_G$. Together with the fact that $f\Aalg$ is a free $e\Aalg$-module of dimension 
$|G:G_e|$ this implies that $\dim_{f\Aalg_G}f\Aalg=|G:G_e|\dim_{e\Aalg_G}e\Aalg$. Furthermore, from 
the last paragraph $\dim_{e\Aalg_G}e\Aalg=|G_e:N_e|$, thus $\dim_{f \Aalg_G}f\Aalg=|G:N_e|\leq |G|$.
Finally, this gives $\dim_k f\Aalg\leq \dim_k f\Aalg_G\cdot|G|$. 
Applying this for all the primitive idempotents $e$ of $\Aalg$ (and thus to all the corresponding primitive 
idempotents $f$ of $\Aalg_G$), we obtain the asserted inequality.

Observe that equality holds iff $|N_e|=1$ for every primitive idempotent $e$ of $\Aalg$. In that case 
for every primitive idempotent $e$ of $\Aalg$, there is no non-identity automorphism in $G$ that fixes 
$e\Aalg$, thus equivalently for every nontrivial ideal $I$ of $\Aalg$ there is no non-identity automorphism 
in $G$ that fixes $I$. This means that equality holds iff  $G$ is semiregular.

Also, equality holds iff $\dim_{f \Aalg_G}f\Aalg=|G|$ for every primitive idempotent $e$ of $\Aalg$. The latter 
condition is equivalent to saying that every component field of $\Aalg_G$ has multiplicity $|G|$ in the 
$\Aalg_G$-module $\Aalg$, this in turn is equivalent to saying that $\Aalg$ is a free $\Aalg_G$-module of 
dimension $|G|$. 
\end{proof}
\fi

\iffull
Using the above Lemma we can decide semiregularity in an efficient way.
\fi

\begin{proposition}\label{pro-regularaut}
Given 
\iffull
a commutative semisimple algebra $\Aalg$ over a finite field $k$, together with 
\fi
a set $\Gamma$ of 
$k$-automorphisms of $\Aalg$. Let $G$ be the group generated by $\Gamma$. In deterministic 
$poly(|\Gamma|, \log |\Aalg|)$ time one can list all the elements of $G$ if $G$ is semiregular, 
or one can find a zero divisor of $\Aalg$ if $G$ is not semiregular.
\end{proposition}
\iffull
\begin{proof}
We first compute $\Aalg_\Gamma$ by linear algebra over $k$. We can assume that $\Aalg$ is a free 
$\Aalg_\Gamma$-module otherwise the algorithm in Lemma \ref{lem-nonfree} finds a zero divisor. 
By Lemma \ref{lem-semireg-ineq} $|G|\ge \dim_{\Aalg_\Gamma}\Aalg=:m$ so try to enumerate 
$(m+1)$ different elements in the group $G$. If we are unable to get that many elements then, 
by Lemma \ref{lem-semireg-ineq}, $G$ is semiregular and we end up with a list of $m$ elements that 
exactly comprise $G$. 

If we do get a set $S$ of $(m+1)$ elements then $G$ is clearly not semiregular. Let $e$ be a primitive 
idempotent of $\Aalg$ such that the subgroup $N_e\leq G$, consisting of automorphisms that fix $e\Aalg$,
is of maximal size. Then from the proof of Lemma \ref{lem-semireg-ineq} we obtain $|G:N_e|\leq m$ which 
means, by pigeon-hole principle, that in the set $S$ there are two different elements 
$\sigma_1, \sigma_2$ such 
that $\sigma:=\sigma_1\sigma_2^{-1}\in N_e$, thus $\sigma$ fixes $e\Aalg$. We now compute 
$\Aalg_\sigma$ and we know from this discussion that $e\Aalg\subseteq\Aalg_\sigma$. Thus we 
get two orthogonal component algebras $e\Aalg_\sigma$ and $(1-e)\Aalg_\sigma$ of 
$\Aalg_\sigma$. We have from the proof of Lemma \ref{lem-semireg-ineq} that 
$e\Aalg_\sigma=(e\Aalg)_\sigma=e\Aalg$ while $(1-e)\Aalg_\sigma=$ 
$((1-e)\Aalg)_\sigma\neq (1-e)\Aalg$ (if $((1-e)\Aalg)_\sigma=(1-e)\Aalg$ then 
$\sigma$ would fix every element in $\Aalg$ and would be a trivial automorphism). 
As a result $\Aalg$ is not a free module over $\Aalg_\sigma$ and hence we can find a 
zero divisor of $\Aalg$ using the method of Lemma \ref{lem-nonfree}. 
\end{proof}
\fi

\iffull
\smallskip
\noindent
{\bf Subgroup $G_\Balg$: }
Let $G$ be a semiregular group of $k$-automorphisms of $\Aalg$ and let $\Balg$ be
a subalgebra of $\Aalg$. We define $G_\Balg$ to be the subgroup of automorphisms of 
$G$ that fix $\Balg$. We give below a Galois theory-like characterization of 
$G_\Balg$.

\begin{proposition}\label{pro-galclosed}
Given a semiregular group $G$ of automorphisms of a commutative semisimple algebra 
$\Aalg$ over a finite field $k$ and a subalgebra $\Balg$ of $\Aalg$ containing 
$\Aalg_G$, one can find a zero divisor in $\Aalg$ in deterministic polynomial time
if $\Balg\not=\Aalg_{G_\Balg}$.
\end{proposition}
\begin{proof}
If $\Aalg$ is a field extension of $k$ then by Galois theory
$\Balg=\Aalg_{G_\Balg}$. If $|k|<(\dim_k\Aalg)^2$ and
$\Aalg$ is not a field then we can find a zero divisor in $\Aalg$
using Berlekamp's deterministic polynomial time algorithm. So
for the rest of the proof we may assume that $|k|\geq (\dim_k\Aalg)^2$ and
then the usual proof of Fact \ref{fac-primitive} gives a deterministic polynomial 
time algorithm for finding a primitive element $x$ of $\Aalg$ over $k$, see 
\cite{GI}. 

Let $|G|=d$. We may assume that the elements $1,x,x^2,\ldots,x^{d-1}$ form a free 
basis of $\Aalg$ over $\Aalg_G$ since otherwise we find a zero divisor in $\Aalg$ 
using the method of Lemma \ref{lem-nonfree}. Let $x^d=\sum_{i=0}^{d-1} a_ix^i$ 
with $a_i\in\Aalg_G$ and let $f(X):=X^d-\sum_{i=0}^{d-1} a_iX^d\in \Aalg_G[X]$.
Obviously $x$ is a root of $f(X)$ and as any $\sigma\in G$ fixes the 
coefficients of $f(X)$ we get that $x^\sigma$ is also a root of $f(X)$.
Again by Lemma \ref{lem-nonfree} we may assume that $\Aalg$ is a $\Balg$-module 
with $\{1,x,\ldots x^{m-1}\}$ as a free basis, where $m:=dim_\Balg\Aalg$. Let
$x^m=\sum_{i=0}^{m-1}b_ix^i$ with $b_i\in\Balg$, thus $x$ is a root of the 
polynomial $g(X):=X^m-\sum_{i=0}^{m-1}b_iX^i\in\Balg[X]$. 

Let us consider $f(X)$ as a polynomial in $\Balg[X]$. As $g(X)$ is monic we can 
apply the usual polynomial division algorithm to obtain polynomials $h(X)$ and 
$r[X]$ from $\Balg(X)$ such that the degree of $h(X)$ is $(d-m)$; the degree of 
$r(X)$ is less than $m$ and $f(X)=g(X)h(X)+r(X)$. We have $r(x)=0$ which together 
with the freeness of the basis $\{1,\ldots,x^{m-1}\}$ implies that $r(X)=0$ and 
$f(X)=g(X)h(X)$. We know from the last paragraph that for all $\sigma\in G$, 
$x^\sigma$ is a root of $g(X)h(X)$. If neither $g(x^\sigma)$ nor $h(x^\sigma)$ is 
zero then we have a pair of zero divisors. If $g(x^\sigma)=0$ then we can perform 
the division of $g(X)$ by $(X-x^\sigma)$ obtaining a polynomial 
$g_1(X)\in \Balg[X]$ with $g(X)=(X-X^\sigma)g_1(X)$ and can then proceed with a new 
automorphism $\sigma'\in G$ and with $g_1(X)$ in place of $g(X)$. In $d$ rounds we 
either find a zero divisor in $\Aalg$ or two disjoint subsets $K, K'$ of $G$ with 
$g(X)=\prod_{\sigma\in K}(X-x^\sigma)$ and $h(X)=\prod_{\sigma'\in K'}(X-x^{\sigma'})$. 
For $\sigma\in K$ let $\phi_\sigma:\Balg[X]\rightarrow \Aalg$ be the homomorphism
which fixes $\Balg$ but sends $X$ to $x^\sigma$. As $g(x^\sigma)=0$, $\phi_\sigma$ 
induces a homomorphism from $\Balg[X]/(g(X))$ to $\Aalg$, which we denote again by 
$\phi_\sigma$. We know that $\phi_1$ is actually an isomorphism 
$\Balg[X]/(g(X))\cong\Aalg$, therefore the maps $\mu_\sigma=\phi_\sigma\circ\phi_1^{-1}$
($\sigma\in K$) are $\Balg$-endomorphisms of $\Aalg$. Note that we can find a zero 
divisor in $\Aalg$ if any $\mu_\sigma$ is not an automorphism, also by Proposition 
\ref{pro-regularaut} we can find a zero divisor in $\Aalg$ if the maps $\mu_\sigma$ 
($\sigma\in K$) generate a non-semiregular group of $\Balg$-automorphisms of $\Aalg$. 
Thus, we can assume that $\mu_\sigma$, for all $\sigma\in K$, generate a semiregular 
group of $\Balg$-automorphisms of $\Aalg$. As $|K|=dim_\Balg\Aalg$ this means, by Lemma
\ref{lem-semireg-ineq}, that the set $\{\mu_\sigma|\sigma\in K\}$ is a group say $H$.
We can as well assume that the group of $k$-automorphisms of $\Aalg$ generated by $G$ and 
$H$ is semiregular, for otherwise we find a zero divisor in $\Aalg$. Again as 
$|G|=dim_k\Aalg$ this means, by Lemma \ref{lem-semireg-ineq}, that $H$ is a subgroup of 
$G$. Thus, by Lemma \ref{lem-semireg-ineq}, $[\Aalg:\Aalg_H]=|H|=|K|=[\Aalg:\Balg]$ which
together with the fact $\Balg\leq \Aalg_H$ gives $\Aalg_H=\Balg$. As $H\leq G_\Balg$ we 
also get $H=G_\Balg$ (if $H<G_\Balg$ then $[\Aalg:\Aalg_H]<$ 
$[\Aalg:\Aalg_{G_\Balg}]\le[\Aalg:\Balg]$ which is a contradiction). Thus, if none of the
above steps yield a zero divisor then $\Balg=\Aalg_{G_\Balg}$.
\end{proof}
\fi

\section{Kummer Extensions and Automorphisms of an Algebra over a
Finite Field}\label{sec-kummer-extn}

In classical field theory a field extension $L$ over $k$ is called a {\em Kummer 
extension} if $k$ has, say, an $r$-th primitive root of unity and $L=k(\sqrt[r]{a})$. 
Kummer extensions are the building blocks in field theory because they have a cyclic 
Galois group. In the previous section we developed a notion of semiregular groups to
mimic the classical notion of Galois groups, now in this section we extend the classical
notion of Kummer extensions to commutative semisimple algebra $\Aalg$ over a finite 
field $k$. The properties of Kummer extensions of $\Aalg$, that we prove in the next 
three subsections, are the reason why we can get polynomial factoring-like results 
without invoking GRH.

\subsection{Kummer-type extensions}

We generalize below several tools and results in field theory, from 
the seminal paper of Lenstra \cite{Len}, to commutative semisimple algebras. 

\smallskip
\noindent
{\bf $k[\zeta_r]$ and $\Delta_r$: }
\iffull
Let $k$ be a finite field and let 
\fi
$r$ be a prime different from $char\ k$. By $k[\zeta_r]$ we denote 
the factor algebra $k[X]/(\sum_{i=1}^{r-1}X^i)$ and $\zeta_r:=X\pmod{\sum_{i=1}^{r-1}X^i}$. 
Then $k[\zeta_r]$ is an $(r-1)$-dimensional $k$-algebra with basis $\{1,\zeta_r,\ldots,\zeta_r^{r-2}\}$ 
and for every integer $a$ coprime to $r$, there exists a unique $k$-automorphism $\rho_a$ of 
$k[\zeta_r]$ which sends $\zeta_r$ to $\zeta_r^a$. Let $\Delta_r$ denote the set of all $\rho_a$'s. 
\iffull

\fi
Clearly, $\Delta_r$ is a group isomorphic to the multiplicative group of integers modulo $r$, therefore 
it is a cyclic group of order $(r-1)$. Note that for $r=2$, we have $\zeta_2=-1$, $\Aalg[\zeta_2]=\Aalg$
and $\Delta_2=\{id\}$.

\smallskip
\noindent
{\bf $\Aalg[\zeta_r]$ and $\Delta_r$: }
Let $\Aalg$ be a commutative semisimple algebra over $k$ then by $\Aalg[\zeta_r]$ we denote 
$\Aalg\otimes_k k[\zeta_r]$. We consider $\Aalg$ as embedded into $\Aalg[\zeta_r]$ via the map 
$x\mapsto x \otimes 1$ and $k[\zeta_r]$ embedded into $\Aalg[\zeta_r]$ via the map 
$x\mapsto 1 \otimes x$. Every element $\rho_a$ of the group $\Delta_r$ can be extended in a unique 
way to an automorphism of $\Aalg[\zeta_r]$ which acts as an identity on $\Aalg$. These extended 
automorphisms of $\Aalg[\zeta_r]$ are also denoted by $\rho_a$ and their group by $\Delta_r$.  
\iffull

\fi
Note that if $\Aalg=\Aalg_1\oplus\ldots\oplus\Aalg_t$ then 
$\Aalg[\zeta_r]=\Aalg_1[\zeta_r]\oplus\ldots\oplus \Aalg_t[\zeta_r]$, thus $\Aalg$'s semisimplicity
implies that $\Aalg[\zeta_r]$ is semisimple as well. We can also easily see the fixed 
points in $\Aalg[\zeta_r]$ of $\Delta_r$ just like Proposition 4.1 of \cite{Len}:

\begin{lemma}\label{lem-delta-fix}
$\Aalg[\zeta_r]_{\Delta_r}=\Aalg$.
\end{lemma}
\iffull
\begin{proof}
Observe that $\Aalg[\zeta_r]$ is a free $\Aalg$-module with basis $\{\zeta_r,\ldots,\zeta_r^{r-1}\}$.
As $r$ is prime this basis is transitively permuted by $\Delta_r$, thus an 
$x=\sum_{i=1}^{r-1}a_i\zeta_r^i\in\Aalg[\zeta_r]$ is fixed by $\Delta_r$ iff $a_i$'s are equal iff 
$x\in\Aalg$. 
\end{proof}
\fi

Consider the multiplicative group $\Aalg[\zeta_r]^*$ of units in $\Aalg[\zeta_r]$. 

\smallskip
\noindent
{\bf Sylow subgroup $\Aalg[\zeta_r]^*_r$: }
Let $\Aalg[\zeta_r]^*_r$ be the $r$-elements of $\Aalg[\zeta_r]^*$. Note that $\Aalg[\zeta_r]^*_r$ is of an 
$r$-power size and is also the $r$-Sylow subgroup of the group $\Aalg[\zeta_r]^*$. Let 
$|\Aalg[\zeta_r]^*_r|=:r^t$.

\smallskip
\noindent
{\bf Automorphism $\omega(a)$: }
Let $a$ be coprime to $r$. Observe that the residue class of $a^{r^{t-1}}$ modulo ${r^t}$ depends 
only on the residue class of $a$ modulo $r$, because map $a\mapsto a^{r^{t-1}}$ corresponds just to 
the projection of the multiplicative group $\Z_{r^t}^*\cong (\Z_{r-1},+)\oplus(\Z_{r^{t-1}},+)$ to the 
first component. This together with the fact that for any $x\in\Aalg[\zeta_r]^*_r$, $x^{r^t}=1$ we get
that the element $x^{a^{r^{t-1}}}$ depends only on the residue class of $a$ modulo $r$. This 
motivates the definition of the map, following \cite{Len}, $\omega(a):x\mapsto x^{\omega(a)}:=$
$x^{a^{r^{t-1}}}$ from $\Aalg[\zeta_r]^*_r$ to itself. Note that we use the term $\omega(a)$ for 
both the above map as well as the residue of $a^{r^{t-1}}$ modulo $r^t$.
\iffull

\fi
Note that the map $\omega(a)$ is an automorphism of the group $\Aalg[\zeta_r]^*_r$ and it commutes 
with all the endomorphisms of the group $\Aalg[\zeta_r]^*_r$.  Also, the map $a\mapsto \omega(a)$ 
is a group embedding $\Z_r^*\rightarrow \Aut(\Aalg[\zeta_r]^*_r)$.

\smallskip
\noindent
{\bf Teichm\"uller subgroup: }
Notice that if $x\in\Aalg[\zeta_r]$ has order $r^u$ then $x^{\omega(a)}=x^{a^{r^{u-1}}}$. Thus, 
$\omega(a)$ can be considered as an extension of the map $\rho_a$ that raised elements of order $r$ to 
the $a$-th power. The elements on which the actions of $\omega(a)$ and $\rho_a$ are the same, for all 
$a$, form the {\em Teichm\"uller subgroup}, $T_{\Aalg,r}$, of $\Aalg[\zeta_r]^*$:
$$T_{\Aalg,r}:=\{x\in\Aalg[\zeta_r]^*_r\ |\ x^{\rho_a}=x^{\omega(a)}\forevery \rho_a\in\Delta_r\}$$
Note that for $r=2$, $T_{\Aalg,2}$ is just the $2$-Sylow subgroup of $\Aalg^*$.

By \cite{Len}, Proposition 4.2, if $\Aalg$ is a field then $T_{\Aalg,r}$ is cyclic .
\iffull
We show in the following lemma 
\fi
that, in our general case, given a witness of 
non-cylicness of $T_{\Aalg,r}$ we can 
compute a zero divisor in $\Aalg$. 

\begin{lemma}\label{lem-cyclicteich}
Given $u,v\in T_{\Aalg,r}$ such that the subgroup generated by $u$ and $v$ is not cyclic, we can find
a zero divisor in $\Aalg$ in deterministic $poly(r,\log |\Aalg|)$ time.
\end{lemma}
\iffull
\begin{proof}
Suppose the subgroup generated by $u$ and $v$ is not cyclic. Then, by Lemma \ref{lem-discretelog}
we can efficiently find a zero divisor $z$, in the semisimple algebra $\Aalg[\zeta_r]$, of the form 
$z=(u^s-v^{s^\prime})$. Next we compute the annihilator ideal $I$ of $z$ in $\Aalg[\zeta_r]$ and 
its identity element $e$, thus $I=e\Aalg[\zeta_r]$. If we can show that $I$ is invariant under $\Delta_r$
then $\Delta_r$ is a group of algebra automorphisms of $I$ which of course would fix the identity element 
$e$ of $I$. Thus, $e$ is in $\Aalg[\zeta_r]_{\Delta_r}$ and hence $e$ is in $\Aalg$ by Lemma 
\ref{lem-delta-fix}, so we have a zero divisor in $\Aalg$. 

Now we show that the annihilator ideal $I=e\Aalg[\zeta_r]$ of $z$ in $\Aalg[\zeta_r]$ is invariant under 
$\Delta_r$. By definition $e$ is an idempotent such that $e(u^s-v^{s^\prime})=0$. Observe that for any
$a\in\{1,\ldots,r-1\}$, we have that $(eu^s)^{\omega(a^{-1})}=(ev^{s^\prime})^{\omega(a^{-1})}$. 
Using this together with the fact that $u^s, v^{s^\prime}\in T_{\Aalg,r}$ we obtain
$e^{\rho_a}(u^s-v^{s^\prime})=(e((u^s)^{\rho_a^{-1}}-(v^{s^\prime})^{\rho_a^{-1}}))^{\rho_a}=
(e((u^s)^{\omega(a^{-1})}-(v^{s^\prime})^{\omega(a^{-1})}))^{\rho_a}=
((eu^s)^{\omega(a^{-1})}-(ev^{s^\prime})^{\omega(a^{-1})})^{\rho_a}=0^{\rho_a}=0$.
Thus, for all $a\in\{1,\ldots,r-1\}$, $e^{\rho_a}\in I$ which means that $I$ is invariant under $\Delta_r$.
\end{proof}
\fi

Now we are in a position to define what we call Kummer extension of an algebra $\Aalg$.

\smallskip
\noindent
{\bf Kummer extension $\Aalg[\zeta_r][\sqrt[s]{c}]$: }
For $c\in\Aalg[\zeta_r]^*$ and a power $s$ of $r$, by $\Aalg[\zeta_r][\sqrt[s]{c}]$ we denote 
the factor algebra $\Aalg[\zeta_r][Y]/(Y^s-c)$ and $\sqrt[s]{c}:=Y \pmod{Y^s-c}$. 

\smallskip
\noindent
{\bf Remark.}
Given $c,c_1\in T_{\Aalg,r}$ such that the order of $c$ is greater than or equal to 
the order of $c_1$ and $c_1$ is not a power of $c$, by Lemma \ref{lem-cyclicteich}, 
we can find a zero divisor in $\Aalg$ in $poly(r,\log |\Aalg|)$ time. Therefore,
the really interesting Kummer extensions are of the form $\Aalg[\zeta_r][\sqrt[s]{c}]$, 
where $c\in T_{\Aalg,r}$ and $\zeta_r$ is a power of $\sqrt[s]{c}$.
\smallskip

Clearly, $\Aalg[\zeta_r][\sqrt[s]{c}]$ is a free $\Aalg[\zeta_r]$-module of rank $s$ with basis 
$\{1,\sqrt[s]{c},\ldots,\sqrt[s]{c}^{s-1}\}$. If $c\in T_{\Aalg,r}$ then $\sqrt[s]{c}$ is an 
$r$-element of $\Aalg[\zeta_r][\sqrt[s]{c}]^*$ and for any integer $a$ coprime to $r$, we now
identify an automorphism of the Kummer extension. Extending \cite{Len}, Proposition 4.3, we 
obtain:

\begin{lemma}\label{lem-rhoa}
Let $c\in T_{\Aalg,r}$. Then we can extend every $\rho_a\in\Delta_r$ to a unique automorphism 
of $\Aalg[\zeta_r][\sqrt[s]{c}]$ that sends $\sqrt[s]{c}$ to $(\sqrt[s]{c})^{\omega(a)}$.
\end{lemma}
\iffull
\begin{proof}
For a $\rho_a\in\Delta_r$ let $\tilde{\rho}_a$ denote the map from $\Aalg[\zeta_r][Y]$ to 
$\Aalg[\zeta_r][\sqrt[s]{c}]$ that fixes $\Aalg$, sends $\zeta_r$ to $\zeta_r^a$ and $Y$ to 
$(\sqrt[s]{c})^{\omega(a)}$. As $c\in T_{\Aalg,r}$, $\tilde{\rho}_a$ maps $c$ to 
$c^{\omega(a)}$ and thus maps $(Y^s-c)$ to zero. This means that $\tilde{\rho}_a$ can be seen 
as an endomorphism of $\Aalg[\zeta_r][\sqrt[s]{c}]$ that sends $\sqrt[s]{c}$ to 
$(\sqrt[s]{c})^{\omega(a)}$. Clearly, $\tilde{\rho}_b\cdot\tilde{\rho}_{b^\prime}$ is the same 
endomorphism as $\tilde{\rho}_{bb^\prime}$ if $b, b^\prime$ are both coprime to $r$. Now as
$\tilde{\rho}_a\cdot\tilde{\rho}_{a^{-1}}=\tilde{\rho}_1$ is the identity automorphism of 
$\Aalg[\zeta_r][\sqrt[s]{c}]$ we get that $\tilde{\rho}_a$ is also an automorphism of
$\Aalg[\zeta_r][\sqrt[s]{c}]$, completing the proof. In the rest of the paper we will use
$\rho_a$ also to refer to the automorphism $\tilde{\rho}_a$.
\end{proof}
\fi

We saw above automorphisms of the Kummer extension $\Aalg[\zeta_r][\sqrt[s]{c}]$ that
fixed $\Aalg$. When $s=r$ we can also identify automorphisms that fix $\Aalg[\zeta_r]$:

\begin{proposition}\label{pro-rootext}
Let $c\in T_{\Aalg,r}$ and $\Delta_r$ be the automorphisms of 
$\Aalg[\zeta_r][\sqrt[s]{c}]$ identified in Lemma \ref{lem-rhoa}. Then there is a 
unique automorphism $\sigma$ of $\Aalg[\zeta_r][\sqrt[r]{c}]$ such that: 
\\(1) $\sigma$ fixes $\Aalg[\zeta_r]$ and maps $\sqrt[r]{c}$ to $\zeta_r\sqrt[r]{c}$. 
\\(2) $\sigma$ commutes with the action of $\Delta_r$. 
\\(3) $\sigma$ is a semiregular automorphism of $\Aalg[\zeta_r][\sqrt[r]{c}]_{\Delta_r}$ 
of order $r$ and $(\Aalg[\zeta_r][\sqrt[r]{c}]_{\Delta_r})_\sigma=\Aalg$.
\end{proposition}
\iffull
\begin{proof}
The map fixing $\Aalg[\zeta_r]$ and mapping $Y$ to $\zeta_r Y$ is clearly an automorphism
of $\Aalg[\zeta_r][Y]/$ $(Y^r-c)$. Thus implying the existence and uniqueness of $\sigma$.

Let $\rho_a\in\Delta_r$ be an automorphism of $\Aalg[\zeta_r][\sqrt[r]{c}]$. Clearly, the 
action of $\sigma$ and $\rho_a$ is commutative on any element $x\in\Aalg[\zeta_r]$. Also,
$(\sqrt[r]{c})^{\sigma\rho_a}=(\zeta_r\sqrt[r]{c})^{\rho_a}=
(\zeta_r\sqrt[r]{c})^{\omega(a)}=\zeta_r^{\omega(a)}(\sqrt[r]{c})^{\omega(a)}=
((\sqrt[r]{c})^{\omega(a)})^\sigma=(\sqrt[r]{c})^{\rho_a\sigma}$. This implies the
commutativity of the actions of $\sigma$ and $\Delta_r$ on $\Aalg[\zeta_r][\sqrt[r]{c}]$.

{}From commutativity it follows that $(\Aalg[\zeta_r][\sqrt[r]{c}]_{\Delta_r})^\sigma=
\Aalg[\zeta_r][\sqrt[r]{c}]_{\Delta_r}$, thus $\sigma$ is an automorphism of 
$\Aalg[\zeta_r][\sqrt[r]{c}]_{\Delta_r}$.
Let $G$ be the group generated by $\Delta_r$ and $\sigma$. Then $G$ is a commutative 
group of order $r(r-1)$. As $\Aalg[\zeta_r][\sqrt[r]{c}]_G=
(\Aalg[\zeta_r][\sqrt[r]{c}]_\sigma)_{\Delta_r}=\Aalg[\zeta_r]_{\Delta_r}=\Aalg$, 
Lemma \ref{lem-semireg-ineq} implies that $G$ is semiregular on 
$\Aalg[\zeta_r][\sqrt[r]{c}]$. But then the subgroup $\Delta_r$ is semiregular as well
and by Lemma \ref{lem-semireg-ineq}: $\dim_k\Aalg[\zeta_r][\sqrt[r]{c}]_{\Delta_r}=
\dim_k\Aalg[\zeta_r][\sqrt[r]{c}]/|\Delta_r|=r\dim_k\Aalg=|(\sigma)|\dim_k\Aalg$.
This again implies that $\sigma$ is a semiregular automorphism of 
$\Aalg[\zeta_r][\sqrt[r]{c}]_{\Delta_r}$.
\end{proof}
\fi

\subsection{$\Aalg$ and the Kummer extension of $\Aalg_\tau$, where 
$\tau\in Aut_k(\Aalg)$}

In this subsection we show how to express $\Aalg[\zeta_r]$ as a Kummer
extension of $\Aalg_\tau$ given a semiregular $\tau\in Aut_k(\Aalg)$
of order $r$. 
The Lagrange resolvent technique of \cite{ro1} remains applicable in our context as 
well and leads to the following:

\begin{lemma}\label{lem-auteigen}
Given a commutative semisimple algebra $\Aalg$ over a finite field $k$, a 
$k$- automorphism $\tau$ of $\Aalg$ of prime order $r\not=char\ k$ and a root 
$\xi\in\Aalg_\tau$ of the cyclotomic polynomial $\frac{X^r-1}{X-1}$. We can find 
in deterministic $poly(r,\log |\Aalg|)$ time a nonzero $x\in\Aalg$ such that 
$x^\tau= \xi x$. 
\end{lemma}
\iffull
\begin{proof}
Observe that if $\xi\in\Aalg$ is a root of $1+X+\ldots+X^{r-1}$
then so is every power $\xi^i\;\;(i=1,\ldots,r-1)$.
Take an element $y\in\Aalg\setminus\Aalg_\tau$ and compute the 
{\em Lagrange-resolvents} for $0\le j\le r-1$: 
$$(y,\xi^j):=\sum_{i=0}^{r-1}\xi^{ij}y^{\tau^i}$$
It is easy to see that $(y,\xi^0)=y+y^\tau+\ldots+y^{\tau^{r-1}}\in\Aalg_\tau$
as $\tau^r=id$, while $\sum_{j=0}^{r-1}(y,\xi^j)=
ry+\sum_{i=1}^{r-1}\sum_{j=0}^{r-1}\xi^{ij}y^{\tau^{i}}=
ry+\sum_{i=1}^{r-1}y^{\tau^{i}}\sum_{j=0}^{r-1}(\xi^{i})^{j}=
ry\not\in \Aalg_\tau$.
It follows that for some $1\le j\le(r-1)$, $(y,\xi^j)\not\in\Aalg_\tau$, fix this 
$j$. In particular, $(y,\xi^j)\not=0$ and taking $l:=(-j)^{-1}\pmod{r}$ we find 
$x:=(y,\xi^j)^l$ is also nonzero as commutative semisimple algebras do not contain 
nilpotent elements. This $x$ is then the element promised in the claim as:
$x^\tau=((y,\xi^j)^\tau)^l=(\xi^{-j}(y,\xi^j))^l=\xi x.$
\end{proof}
\fi

We now proceed to describe an algorithm that given a $k$-automorphism $\tau$ of 
$\Aalg$ of prime order $r$, expresses $\Aalg[\zeta_r]$ as a Kummer extension of 
$\Aalg_\tau$.

\smallskip
\noindent
{\bf Embedding $Aut_k(\Aalg)$ in $Aut_k(\Aalg[\zeta_r])$: }
Given a semiregular automorphism $\tau$ of $\Aalg$ we extend $\tau$ to an 
automorphism of $\Aalg[\zeta_r]$ by letting $\zeta_r^\tau:=\zeta_r$. It is easy to 
see that the extension (denoted again by $\tau$) is a semiregular automorphism of 
$\Aalg[\zeta_r]$ as well and it commutes with $\Delta_r$.

Application of Lemma~\ref{lem-auteigen}, techniques from
\cite{Len} and a careful treatment of cases when we find zero divisors, give the following.

\begin{proposition}\label{pro-teichresolv}
Given a commutative semisimple algebra $\Aalg$ over a finite field $k$ together 
with a semiregular $k$-automorphism $\tau$ of $\Aalg$ of prime order 
$r\not=char\ k$, we can find in deterministic $poly(\log |\Aalg|)$ time an 
element $x\in T_{\Aalg,r}$ such that $x^\tau=\zeta_r x$.

Any such $x$ satisfies $c:=x^r\in T_{\Aalg_\tau,r}$ and defines an isomorphism 
$\phi:\Aalg_\tau[\zeta_r][\sqrt[r]{c}]\cong\Aalg[\zeta_r]$ 
which fixes $\Aalg_\tau[\zeta_r]$. Also $\phi$ commutes with the action of 
$\Delta_r$, therefore inducing an isomorphism 
$(\Aalg_\tau[\zeta_r][\sqrt[r]{c}])_{\Delta_r}\cong\Aalg$. 
\end{proposition}
\iffull
\begin{proof}
The proof idea is to first apply Lemma \ref{lem-auteigen} to find a nonzero 
$x\in\Aalg[\zeta_r]$ such that $x^\tau=\zeta_r x$. Note that this $x$ maybe a zero 
divisor of $\Aalg[\zeta_r]$, in that case we intend to decompose $\Aalg[\zeta_r]$ as 
much as possible and apply Lemma \ref{lem-auteigen} to each of these components. 
This process is repeated till it yields an $y\in\Aalg[\zeta_r]^*$ such that 
$y^\tau=\zeta_r y$. Secondly, this $y$ is used to form the $x$ and $\phi$ 
as promised in the claim.

We maintain: a decomposition of the identity element $1=1_{\Aalg[\zeta_r]}=1_\Aalg$ into 
orthogonal idempotents $e, f$ that are fixed by $\tau$; and an element 
$y\in (f\Aalg[\zeta_r])^*$ such that $y^\tau=\zeta_r y$ (for $f=0$ we define 
$(f\Aalg[\zeta_r])^*$ as $(0)$). Initially, we take $e=1,\;f=0,\;y=0$. Since $\tau$ is 
semiregular its restriction to $e\Aalg[\zeta_r]$ has to be nontrivial (as long as $e\not=0$) 
and hence of prime order $r$. Therefore we can apply Lemma \ref{lem-auteigen} with 
$\xi=e\zeta_r$ to find a nonzero $x\in e\Aalg[\zeta_r]$ such that 
$x^\tau=(e\zeta_r)x=\zeta_r x$. Now compute the identity element $e_1$ of 
$x\Aalg[\zeta_r]$ (which is an ideal of $e\Aalg[\zeta_r]$). Note that $x\Aalg[\zeta_r]$ is 
invariant under $\tau$ since for all $z\in\Aalg[\zeta_r]$, $(xz)^\tau=x^\tau z^\tau=
\zeta_r x z^\tau\in x\Aalg[\zeta_r]$. This makes $\tau$ an automorphism of $x\Aalg[\zeta_r]$ 
and so $\tau$ fixes the identity element $e_1$. We could now replace $e$ with $(e-e_1)$, $f$ 
with $(f+e_1)$, $y$ with $(x+y)$ and repeat the above steps. Note that the above one iteration 
decomposed $e\Aalg[\zeta_r]$ into orthogonal components $(e-e_1)\Aalg[\zeta_r]$ and 
$e_1\Aalg[\zeta_r]$ and thus the procedure has to stop in at most $\dim_k\Aalg[\zeta_r]$ 
rounds with $e=0$.

So far we have found an element $y\in\Aalg[\zeta_r]^*$ with $y^\tau=\zeta_r y$. Define
$|\Aalg[\zeta_r]_r^*|=:r^t$, $\ell:=|\Aalg[\zeta_r]^*|/r^t$ and $m:=(-\ell)^{-1}\pmod{r}$.
Note that $\ell$ can be calculated from the sizes of the simple components of $\Aalg[\zeta_r]$ 
which in turn can be easily computed by using the standard distinct degree factorization of
polynomials over finite fields. Thus, we can compute the element $z:=y^{\ell m}$. By the 
definition of $\ell$ and $y$, $z\in\Aalg[\zeta_r]_r^*$ and $z^\tau=\zeta_r^{\ell m} z=
\zeta_r^{-1}z$. Next compute the element 
$x=\prod_{b=1}^{r-1}(z^{\omega(b)})^{\rho_b^{-1}}$. Note that for all $\rho_a\in\Delta_r$, 
$x^{\rho_a}=\prod_{b=1}^{r-1}(z^{\omega(a^{-1}b)\omega(a)})^{\rho_{a^{-1}b}^{-1}}=
x^{\omega(a)}$, whence $x\in T_{\Aalg,r}$. Also, as $\tau$ commutes with $\Delta_r$ we have
$x^\tau=\prod_{b=1}^{r-1}((\zeta_r^{-1}z)^{\omega(b)})^{\rho_b^{-1}}=
x\cdot\prod_{b=1}^{r-1}((\zeta_r^{-1})^{\omega(b)})^{\rho_b^{-1}}=
(\zeta_r^{-1})^{r-1}x=\zeta_r x$. Finally, we define the $c$ as $x^r$. {}From the properties of $x$,
$c\in\Aalg[\zeta_r]_\tau=\Aalg_\tau[\zeta_r]$ and hence $c\in T_{\Aalg_\tau,r}$.

Let us define the map $\phi$ from $\Aalg_\tau[\zeta_r][\sqrt[r]{c}]$ to $\Aalg[\zeta_r]$ as the one 
that sends $\sqrt[r]{c}$ to $x$ and fixes $\Aalg_\tau[\zeta_r]$. It is obvious from $c=x^r$ that 
$\phi$ is a homomorphism. If $\phi$ maps an element $\sum_{i=0}^{r-1}a_i(\sqrt[r]{c})^i$
to zero then $\sum_{i=0}^{r-1}a_i x^i=0$. Applying $\tau$ on this $j$ times gives 
$\sum_{i=0}^{r-1}a_i \zeta_r^{ij}x^i=0$ (remember $\tau$ fixes $\Aalg_\tau[\zeta_r]$ and 
hence $a_i$'s). Summing these equations for all $0\le j\le (r-1)$ we get $a_0=0$, as $x$ is invertible
this means that $\phi$ maps $\sum_{i=1}^{r-1}a_i x^{i-1}$ to zero. We can now repeat the 
argument and deduce that $a_i$'s are all zero, thus $\phi$ is injective.
Using that $x\in T_{\Aalg,r}$, it is also straightforward to verify that 
$\phi$ commutes with $\Delta_r$ (viewed as automorphisms of $\Aalg[\zeta_r][\sqrt[r]{c}]$). 
Thus it remains to show that $\phi$ is surjective. To this end let $\Balg$ denote the image of 
$\phi$. Then $\Balg$ is the subalgebra of $\Aalg[\zeta_r]$ generated by $\Aalg_\tau[\zeta_r]$ and 
$x$, thus $\Balg$ is $\tau$-invariant. Suppose we can show $\tau$ semiregular on $\Balg$. 
Then by Lemma \ref{lem-semireg-ineq}, $\dim_k \Balg=r\dim_k \Balg_\tau$, this together with 
$\Balg_\tau$ containing $\Aalg_\tau[\zeta_r]$ and the injectivity of $\phi$ means that $\dim_k \Balg\ge 
r\dim_k \Aalg_\tau[\zeta_r]=r\dim_k \Aalg[\zeta_r]_\tau$ which is further equal to 
$\dim_k\Aalg[\zeta_r]$ as $\tau$ is semiregular on $\Aalg[\zeta_r]$. Thus, 
$\dim_k \Balg\ge \dim_k\Aalg[\zeta_r]$ which obviously means that $\phi$ is indeed surjective.

It remains to prove the semiregularity of $\tau$ on $\Balg$. Assume for contradiction that $I$ is a 
nonzero ideal of $\Balg$ such that $\tau$ fixes $I$ and $e$ be the identity element of $I$. Then 
$(ex)^\tau=ex$. On the other hand, as $e^\tau=e$ and $x^\tau=\zeta_r x$, we have 
$(ex)^\tau=\zeta_r ex$. Combining the two equalities we obtain that $(ex)(\zeta_r-1)=0$. Note that
if $r=2$ then $char\ k>2$ and $(\zeta_r-1)$ is not a zero divisor and if $r>2$ then $\Aalg[\zeta_r]$ 
is a free $\Aalg$-module with basis $\{1,\ldots,\zeta_r^{r-2}\}$. Thus, $x(\zeta_r-1)$ is invertible
in all cases, implying $e=0$ which is a contradiction. Thus $\tau$ is indeed semiregular on 
$\Balg$ completing the proof that $\phi$ is an isomorphism.  
\end{proof}
\fi

\subsection{Zero Divisors using Noncyclic Groups: Proof of Application 2}

In this part we prove Application 2 by proving the following stronger result.

\begin{theorem}\label{thm-cyclicaut}
Given a commutative semisimple algebra $\Aalg$ over a finite field $k$ together with a 
noncyclic group $G$ of $k$-automorphisms of $\Aalg$ (in terms of generators), one can find 
a zero divisor in $\Aalg$ in deterministic polynomial time.
\end{theorem}
\begin{proof}
Notice that since $G$ is noncyclic, the algebra $\Aalg$ is certainly not a field and zero 
divisors do exist. We assume that $G$ is semiregular otherwise we can efficiently find a zero 
divisor in $\Aalg$ by Proposition \ref{pro-regularaut}. We can also assume that $|G|$ is not 
divisible by $char\ k$ otherwise $char\ k\le |G|\le \dim_k \Aalg$ and Berlekamp's deterministic 
algorithm for polynomial factoring can be used to find all the simple components of $\Aalg$.

As $G$ is a small group of size $\dim_k \Aalg$, we can list
all its elements of prime order. 
The proof now proceeds by analyzing the Sylow subgroups of $G$ and showing them all 
cyclic unless they yield a zero divisor of $\Aalg$. 
For every prime divisor $r$ of $|G|$ let $\Pi_r$ be the set
of elements of $G$ of order $r$ and let $P_r$ be an
$r$-Sylow subgroup of $G$. For every $\sigma\in\Pi_r$
we can use Proposition \ref{pro-teichresolv} to compute an 
element $x_\sigma\in T_{\Aalg,r}$ with 
$x_\sigma^{\sigma}=\zeta_r x_\sigma$. 
Let $H_r$ be the subgroup of $T_{\Aalg,r}$ generated 
by $\{x_\sigma|\sigma\in\Pi_r\}$. 

We can assume $H_r$ to be cyclic or else we can find a 
zero divisor in $\Aalg$ by Lemma \ref{lem-cyclicteich}. So choose an element 
$x\in\{x_\sigma|\sigma\in\Pi_r\}$ such that $x$ is a generator of $H_r$. 
Now for any 
$\sigma\in G$, as $x^\sigma$ is again in $T_{\Aalg,r}$, we can assume $x^\sigma\in H_r$ for 
otherwise we can find a zero divisor by Lemma \ref{lem-cyclicteich}. Thus, $H_r$ is $G$-invariant 
and $G$ acts as a group of automorphisms of $H_r$. As every element of 
$P_r$ of order $r$ moves some element in $H_r$, there is no nontrivial 
element of $P_r$ acting trivially on $H_r$, thus $P_r$ 
intersects trivially with the kernel $K_r$ of the restriction homomorphism $G\rightarrow Aut(H_r)$. 
Since $H_r$ is cyclic, 
its automorphism group is Abelian. The last two observations imply that
$G/K_r$ is an Abelian group 
with a natural embedding of $P_r\rightarrow G/K_r\cong Aut(H_r)$. Thus the normal series
$K_r\lhd G$ can be 
refined to $K_r\unlhd N_r\lhd G$ such that $|P_r|=|G/N_r|$. Since we have this for every $r$ dividing $|G|$, 
it follows that $G$ is a direct product of its Sylow subgroups. Also, as each $P_r$ is
Abelian, $G$ is 
Abelian. Moreover, since the automorphism group of a cyclic group of odd prime-power order is cyclic, 
$Aut(H_r)$ is cyclic and finally $P_r$ is cyclic, for every odd prime $r||G|$.

It remains to show that we can find a zero divisor efficiently if the $2$-Sylow subgroup $P_2$ of $G$ 
is not cyclic. To this end we take a closer look at the subgroup $H_2$ constructed for the prime $r=2$ 
by the method outlined above. It is generated by an element $x$, contains $-1$, and $P_2$ acts 
faithfully as a group of automorphisms of $H_2$. If $|H_2|=2^k$ then $Aut(H_2)\cong \Z_{2^k}^*$. 
As $P_2$ injectively embeds in $Aut(H_2)$ and $P_2$ is noncyclic we get that $\Z_{2^k}^*$ is 
noncyclic, implying that $k>2$ and structurally $\Z_{2^k}^*$ is the direct product of the cyclic groups 
generated by $(-1)$ and $(5)$ modulo $2^k$ respectively. Now any noncyclic subgroup of such a 
$\Z_{2^k}^*$ will have the order $2$ elements: $(-1)$ and $5^{2^{k-3}}\equiv (2^{k-1}+1)$. 
Thus, $P_2$ has the maps $\sigma_1:x\mapsto x^{-1}$ and 
$\sigma_2:x\mapsto x^{2^{k-1}+1}=-x$. Since $\sigma_1$ and $\sigma_2$ commute, 
$\Aalg_{\sigma_1}$ is $\sigma_2$-invariant. As the group $(\sigma_1, \sigma_2)$ is of size $4$ while 
the group $(\sigma_1)$ is only of size $2$ we get by the semiregularity of $G$ that the restriction of 
$\sigma_2$ to $\Aalg_{\sigma_1}$ is not the identity map. Hence, by Proposition \ref{pro-teichresolv} 
we can find an element $y\in T_{\Aalg_{\sigma_1},2}$ such that $y^{\sigma_2}=-y$. We can assume 
that the subgroup of $\Aalg^*$ generated by $x$ and $y$ is cyclic as otherwise we find a zero divisor 
by Lemma \ref{lem-cyclicteich}. However, as $x\not\in\Aalg_{\sigma_1}$ while $y\in\Aalg_{\sigma_1}$, 
it can be seen that: $(x,y)$ is a cyclic group only
if $y\in H_2^2$ (i.e. $y$ is square of an element in $H_2$). But this is a contradiction because 
$\sigma_2$ fixes $H_2^2$. This finishes the proof.
\end{proof}

Now we can give a proof of Application 2.
Let $r$ be a positive integer such that the multiplicative group $\Z_r^*$ is noncyclic and let $\phi_r(x)$ 
be the $r$-th cyclotomic polynomial. We can 
assume $r$ to be coprime to $char\ k$ as otherwise we factor $\phi_r(x)$ simply
by using Berlekamp's algorithm for polynomial factoring. Define
$\Aalg:=k[x]/(\phi_r(x))$, it is clearly a commutative semisimple algebra 
of dimension $\phi(r)$ over $k$. Moreover, if $\zeta_r\in\overline{k}$ is a 
primitive $r$-th root of unity then: 
$\phi_r(x)=\prod_{i\in\Z_r^*}(x-\zeta_r^i)$.
This implies that for any $i\in\Z_r^*$, $\phi_r(x)|\phi_r(x^i)$ and 
if for a $g(X)\in k[X]$, $\phi_r(x)|g(x^i)$ then $\phi_r(X)|g(X)$ as well.
In other words for any $i$ coprime to $r$ the map $\rho_i:x\rightarrow x^i$ is 
a $k$-automorphism of $\Aalg$. Consider the group $G:=\{\rho_i|i\in\Z_r^*\}$, it
is clearly isomorphic to the multiplicative group $\Z_r^*$, which is noncyclic for our $r$. 
Thus, $G$ is noncyclic and we can find 
a zero divisor $a(x)\in\Aalg$ by Theorem \ref{thm-cyclicaut}. Finally, the gcd
of $a(x)$ and $\phi_r(x)$ gives a nontrivial factor of $\phi_r(x)$.

Rational polynomials known to have small but noncommutative
Galois groups also emerge in various branches of mathematics
and its applications. 
For example, the six roots
of the polynomial $F_j(X)=(X^2-X+1)^3-\frac{j}{2^8}X^2(X-1)^2$ are
the possible parameters $\lambda$ of the elliptic curves
 from the {\em Legendre family} $E_\lambda$ 
having prescribed $j$-invariant $j$, see \cite{huse}. (Recall that the 
curve $E_\lambda$ is defined by the equation
$Y^2=X(X-1)(X-\lambda)$.)
The Galois group of $F_j(X)$ is $S_3$, whence
Theorem~\ref{thm-cyclicaut} gives a partial factorization
of the polynomial $F_j(X)$ modulo $p$ 
where $p$ is odd and $j$ is coprime to $p$.

\iffull

\subsection{Extending Automorphisms of $\Aalg_\tau$ to $\Aalg$, where 
$\tau\in Aut_k(\Aalg)$}\label{sec-glue}

\begin{lemma}\label{lem-extendaut}
Given a commutative semisimple algebra $\Aalg$ over a finite field $k$, a $k$- automorphism $\tau$ 
of $\Aalg$ and a $k$-automorphism $\mu$ of $\Aalg_\tau$. Assume that the order of $\tau$ is 
coprime to $char\ k$. Then in deterministic $poly(\log |\Aalg|)$ time we can compute either a zero 
divisor in $\Aalg$ or a $k$-automorphism $\mu'$ of $\Aalg$ that extends $\mu$ such that 
$\Aalg_{\mu'}=(\Aalg_\tau)_\mu$.
\end{lemma}
\begin{proof}
Suppose that the order of $\tau$ is $r_1\cdots r_t$, where
$r_i$'s are primes (not necessarily distinct). Cleary it is sufficient to show how to extend
$\mu$ from $\Aalg_{\tau^{r_1\cdots r_{i-1}}}$ to
$\Aalg_{\tau^{r_1\cdots r_{i}}}$ (or find a zero divisor
during the process).
We can therefore assume that the order of $\tau$ is a prime $r$.
We may also assume that both 
$\tau$ and $\mu$ are semiregular since otherwise
we can find a zero divisor in $\Aalg$ by Proposition~\ref{pro-regularaut}.
We work 
in the algebra $\Aalg[\zeta_r]$. 
We extend $\tau$ to $\Aalg[\zeta_r]$ and $\mu$ to 
$\Aalg_\tau[\zeta_r]$ in the natural way. By Proposition \ref{pro-teichresolv}, we can efficiently 
find $x\in T_{\Aalg,r}$ such that $x^\tau=\zeta_rx$. Clearly, $c:=x^r\in T_{\Aalg_\tau,r}$ and  
$c^\mu\in T_{\Aalg_\tau,r}$. The elements $c$ and $c^\mu$ have the same 
order.
If $c^\mu$ is not in the cyclic group generated by $c$ then by Lemma 
\ref{lem-cyclicteich}, we can find a zero divisor in $\Aalg$. So assume that $c^\mu$ is in the cyclic 
group of $c$, in which case find an integer $j$ coprime to $r$ such that $c^\mu=c^j$ using 
Lemma \ref{lem-discretelog}. Note that by Lemma \ref{lem-cyclicteich},
we can also find a zero divisor in $\Aalg$ in the case when $\zeta_r$ 
is not a power of $c$, so assume that $\zeta_r=c^\ell$ and compute this integer $\ell$. Then 
$\zeta_r=\zeta_r^\mu=(c^\ell)^\mu=(c^\mu)^\ell=c^{j\ell}=\zeta_r^j$, and hence 
$j\equiv 1 \pmod{r}$. We set $x':=x^j$. As $x^\tau=\zeta_r x$ and $x'^\tau=\zeta_r x'$, by the 
proof of Proposition \ref{pro-teichresolv}, there are isomorphism maps 
$\phi:\Aalg_\tau[\zeta_r][\sqrt[r]{c}]\rightarrow \Aalg[\zeta_r]$ and 
$\phi':\Aalg_\tau[\zeta_r][\sqrt[r]{c^\mu}]\rightarrow \Aalg[\zeta_r]$ sending $\sqrt[r]{c}$ to 
$x$ and $\sqrt[r]{c^\mu}$ to $x'$ respectively; both fixing $\Aalg_\tau[\zeta_r]$. We can naturally 
extend $\mu$ to an isomorphism map $\mu'':\Aalg_\tau[\zeta_r][\sqrt[r]{c}]\rightarrow 
\Aalg_\tau[\zeta_r][\sqrt[r]{c^\mu}]$. Then the composition map $\mu':=\phi'\circ\mu''\circ\phi^{-1}$
is an automorphism of $\Aalg[\zeta_r]$ whose restriction to $\Aalg_\tau[\zeta_r]$ is $\mu$. As $\mu''$, 
$\phi$ and $\phi'$ commute with $\Delta_r$, so does $\mu'$. Therefore $\Aalg=
\Aalg[\zeta_r]_{\Delta_r}$ is $\mu'$-invariant and we have the promised $k$-automorphism of 
$\Aalg$.
\end{proof}

\subsection{Zero Divisors using Galois Groups: Proof of Application 3}

If the input polynomial $f(x)\in\Q[x]$ has a ``small'' Galois group then can we
factor $f(x)$ modulo a prime $p$? This question was studied in \cite{ro4} and an
algorithm was given assuming GRH. In this subsection we give a GRH-free version.
We start with the following unconditional and generalized version of Theorem 3.1. 
in \cite{ro4}:

\begin{theorem}\label{thm-galalg}
Assume that we are given a semiregular group $G$ of automorphisms of a
commutative semisimple algebra $\Aalg$ over a finite field $k$ with 
$\Aalg_G=k$ 
and a nonzero ideal $\Balg$ (with $k$ embedded) of a subalgebra of $\Aalg$. Then in deterministic 
$poly(\log |\Aalg|)$ time we can either find a zero divisor in $\Balg$ or a 
semiregular $k$-automorphism $\sigma$ of $\Balg$ of order $\dim_k\Balg$.
\end{theorem}
\noindent
{\bf Remark.} Here $\Balg$ is an {\it ideal} of a subalgebra of $\Aalg$, thus it is not assumed that 
$1_\Aalg\in\Balg$.
\begin{proof}
The idea of the algorithm is to find a nontrivial ideal $I$ of $\Aalg$ and then 
reduce the problem to the smaller instance $I$.

If $G$ is noncyclic then using Theorem \ref{thm-cyclicaut} we can find a nontrivial 
ideal $I$ of $\Aalg$. If $G$ is cyclic then using Proposition \ref{pro-galclosed} we can 
find either a nontrivial ideal $I$ of $\Aalg$ or a subgroup $H$ of $G$ with 
$\Balg=\Aalg_H$. In the latter case $H$ is trivially a normal subgroup of $G$ and the 
restriction of any generator $\sigma$ of $G$ will generate a semiregular group, of 
$k$-automorphisms of $\Balg$, isomorphic to $G/H$. Thus, we get a semiregular 
$k$-automorphism of $\Balg$ of order $|G/H|=\dim_k\Balg$.

Let us assume we have a nontrivial ideal $I$ of $\Aalg$. Then, using the method of Lemma 
\ref{lem-refine}, we find an ideal $J$ of $\Aalg$ such that the ideals 
$\{J^\sigma|\sigma\in G\}$ are pairwise orthogonal or equal. By the hypothesis $\Aalg_G=k$, 
$G$ acts transitively on the minimal ideals of $\Aalg$, thus the group 
$G_1:=\{\sigma\in G|J^\sigma=J\}$ acts semiregularly on $J$ and for coset representatives 
$C$ of $G/G_1$: $\Aalg=\oplus_{\sigma\in C}J^\sigma$. Also, note that for all $\sigma\in C$ 
the conjugate subgroup $G_1^\sigma:=\sigma^{-1}G_1\sigma$ acts semiregularly on $J^\sigma$. 
We can find a zero divisor in $\Balg$ if the projection of $\Balg$ to some $J^\sigma$ is 
neither the zero map nor injective. Thus we assume that there is an ideal $J^\sigma$ such 
that the projection of $\Aalg$ onto $J^\sigma$ injectively embeds $\Balg$. In that case we 
reduce our original problem to the smaller instance -- $J^\sigma$ instead of $\Aalg$, 
$G_1^\sigma$ instead of $G$ and the embedding of $\Balg$ instead of $\Balg$ -- and apply
the steps of the last paragraph. 
\end{proof}

The following Corollary gives the proof of a slightly stronger version of Application 3.

\begin{corollary}\label{cor-smallgalpol}
Let $F(X)\in\Z[X]$ be a polynomial irreducible over $\Q$ with Galois group of size $m$; let $L$ be 
the maximum length of the coefficients of $F(X)$; let $p$ be a prime not dividing the discriminant 
of $F(X)$; let $f(X):=F(X)\pmod{p}$; and 
let $g(X)$ be a non-constant divisor of $f(X)$ in $\F_p[X]$. Then by a 
deterministic $poly(m,L,\log p)$ time algorithm we can find 
either a nontrivial factor of $g(X)$ or an 
automorphism of order $\deg g$ of the algebra $\F_p[x]/(g(x))$.
\end{corollary}
\begin{proof}
The assumption on the discriminant implies that the leading coefficient of $F(X)$ is not divisible by
$p$, and wlog we can assume $F(X)$ to be monic. Also assume that $p> m^4$ as otherwise we 
can use Berlekamp's deterministic algorithm for factoring $f(x)$ completely. Now using the 
algorithm of Theorem 5.3. of \cite{ro4}, we compute an algebraic integer $\alpha:=x\pmod{H(x)}$ 
generating the splitting field $\Q[x]/(H(x))$ of $F(X)$ such that the discriminant of the minimal 
polynomial $H(X)$ of $\alpha$ is not divisible by $p$. Define $\Aalg:=\Z[\alpha]/(p)$ and using the 
method described in Section 4 of \cite{ro4}, we efficiently compute a group $G$ of automorphisms 
of $\Aalg$ which is isomorphic to the Galois group of $\alpha$ over rationals. 

Let $\beta\in\Q[x]/(H(x))$ be a root of $F(X)$. Then $\beta=\sum_{i=0}^{m-1} a_i\alpha^i$ for 
some $a_i\in\Q$. {}From Proposition 13 of Chapter 3 in \cite{Lang}, for every $0\leq i <m$, $a_i$ can 
be written in the form $a_i=r_i/q_i$, where $r_i,q_i\in \Z$ and $q_i$ is coprime to $p$. Compute 
$t_i\in \Z$ with $t_iq_i\equiv 1\pmod{p}$. Then $\beta':=\sum_{i=0}^{m-1}r_it_i\alpha^i$ is in 
$\Z[\alpha]$ and the minimal polynomial of the element $\overline{\beta}:=\beta'\pmod{p}\in\Aalg$ 
is $f(X)$.
Let $\Calg$ be the subalgebra $\F_p[\overline{\beta}]$ contained in $\Aalg$. 
Notice that $\Calg$ is 
isomorphic to the algebra $\F_p[x]/(f(x))$. Let $\Balg$ be the ideal of
$\Calg$ generated by $f(\overline\beta)/g(\overline\beta)$. Then
$\Balg$ is isomorphic to the algebra $\F_p[x]/(g(x))$
and hence a zero divisor of $\Balg$ will give us a factor of 
$g(X)$. So we run the algorithm described in Theorem \ref{thm-galalg} on $G,\Aalg,\Balg$ 
and get either a factor of $g(X)$ or an automorphism of $\Balg$ of order $\dim_{\F_p} \Balg$ , thus 
finishing the proof.
\end{proof}

\section{Finding Automorphisms of Algebras via Kummer Extensions}\label{sec-findaut}

In this section we complete the proof of our main Theorem, i.e. given a commutative semisimple 
algebra $\Aalg$ over a finite field $k$ we can unconditionally find a nontrivial $k$-automorphism 
of $\Aalg$ in deterministic subexponential time. The proof involves computing tensor powers of
$\Aalg$, whose automorphisms we know, and then {\em bringing down} those automorphisms to $\Aalg$.
Before embarking on the proof we need to first see how to bring down automorphisms using Kummer 
extensions; and define notions related to tensor powers of $\Aalg$.

\subsection{Bringing Down Automorphisms of $\Dalg$ to $\Aalg\le\Dalg$}\label{sec-bring-down}

We do this by using Kummer extensions, so we first show how to embed a Kummer 
extension of $\Aalg$ into the cyclotomic extension of $\Dalg$.

\begin{lemma}\label{lem-primeext}
Let $\Aalg\leq\Dalg$ be commutative semisimple algebras over a finite field $k$ and let 
$r\not=char\ k$ be a prime. Then for any $x\in T_{\Dalg,r}\setminus\Aalg[\zeta_r]$
satisfying $c:=x^r\in\Aalg[\zeta_r]$, there is a unique ring homomorphism 
$\phi:\Aalg[\zeta_r][\sqrt[r]{c}]\rightarrow\Dalg[\zeta_r]$ that fixes $\Aalg[\zeta_r]$, 
maps $\sqrt[r]{c}$ to $x$ and:
\\(1) $\phi$ commutes with the action of $\Delta_r$, thus
$\phi(\Aalg[\zeta_r][\sqrt[r]{c}]_{\Delta_r})\subseteq\Dalg$.
\\(2) $\phi$ is injective if and only if its restriction to 
$\Aalg[\zeta_r][\sqrt[r]{c}]_{\Delta_r}$ is injective.
\\(3) If $\phi$ is not injective then we can find a zero divisor of $\Dalg$ in 
deterministic polynomial time .
\end{lemma}
\begin{proof}
The existence and uniqueness of the homomorphism $\phi$ are obvious: 
the map from $\Aalg[\zeta_r][X]$ to $\Dalg[\zeta_r]$ which sends $X$
to $x$ factors through $\Aalg[\zeta_r][\sqrt[r]{c}]$.

As $x\in T_{\Dalg,r}$, for every $\rho_a\in\Delta_r$ we have
$\phi((\sqrt[r]{c})^{\rho_a})=\phi((\sqrt[r]{c})^{\omega(a)})=x^{\omega(a)}=
(\phi(\sqrt[r]{c}))^{\rho_a}$. On the other hand, for every $u\in \Aalg[\zeta_r]$ we 
have $\phi(u)^{\rho_a}=u^{\rho_a}=\phi(u^{\rho_a})$. As $\Aalg[\zeta_r]$ and 
$(\sqrt[r]{c})$ generate $\Aalg[\zeta_r][\sqrt[r]{c}]$, the two equalities above prove 
that $\phi$ commutes with the action of $\Delta_r$. As a consequence, 
$\phi(\Aalg[\zeta_r][\sqrt[r]{c}]_{\Delta_r})\subseteq\Dalg[\zeta_r]_{\Delta_r}=\Dalg$. 

Since the elements $\zeta_r^0,\ldots,\zeta_r^{r-2}$ form a free basis of 
$\Dalg[\zeta_r]$ as a $\Dalg$-module,  the subspaces $\zeta_r^i\Dalg$ of 
$\Dalg[\zeta_r]$ ($i=0,\ldots,r-2$) are independent over $k$. This means the images 
$\phi(\zeta_r^i(\Aalg[\zeta_r][\sqrt[r]{c}]_{\Delta_r}))$ are independent as well thus, 
$\dim_k\phi(\Aalg[\zeta_r][\sqrt[r]{c}])=(r-1)\dim_k\phi($
$\Aalg[\zeta_r][\sqrt[r]{c}]_{\Delta_r})$. This together with the fact 
$\dim_k\Aalg[\zeta_r][\sqrt[r]{c}]=(r-1)\dim_k\Aalg[\zeta_r][\sqrt[r]{c}]_{\Delta_r}$ 
means that $\phi$ is injective if and only if its restriction to 
$\Aalg[\zeta_r][\sqrt[r]{c}]_{\Delta_r}$ is.

To see the last assertion assume that $\phi$, and hence its restriction to 
$\Calg:=\Aalg[\zeta_r][\sqrt[r]{c}]_{\Delta_r}$, is not injective. We compute the 
kernel $I$ of $\phi|_\Calg$, clearly $I$ is a nonzero ideal of $\Calg$. Let $\sigma$ 
be the semiregular $k$-automorphism of $\Calg$ investigated in Proposition 
\ref{pro-rootext}, which also tells us that $\dim_k\Calg=r\dim_k\Aalg$. Assume that 
$\phi(\Calg)=:\Dalg'$. We compute $J:=\{u\in\Calg|uI=0\}$, the ideal complementary to 
$I$ so that $\Calg=I\oplus J$. Note that by the definition of $I$, the restriction of 
$\phi$ to $J$ yields an isomorphism $J\cong\Dalg'$. Hence finding a zero divisor in 
$J$ implies finding a zero divisor in $\Dalg$. Let $e_J$ be the identity 
element of $J$, then as $\phi$ fixes $\Aalg$, for all $a\in\Aalg$, 
$a=\phi(a)=\phi(e_Ja)$, in other words $\phi$ induces an isomorphism 
$e_J\Aalg\cong\Aalg$. Using this we now show that the action of $\sigma$ on $J$ yields 
a zero divisor in $J$. 
 
Firstly, we claim that for all $1\le i\le (r-1)$, $J\not=J^{\sigma^i}$. Suppose for 
some $1\le i\le (r-1)$, $J^{\sigma^i}=J$ and $\sigma^i$ fixes $J$, then 
$J\subseteq\Calg_{\sigma^i}=\Aalg$. This together with the fact that $\phi^{-1}$ 
injectively embeds $\Aalg$ in $J$ gives $J=\Aalg$, which implies that 
$\phi(\Calg)=\Aalg$, thus $\phi(\Aalg[\zeta_r][\sqrt[r]{c}])=\phi(\Calg[\zeta_r])=
\Aalg[\zeta_r]$ contradicting $x\not\in\Aalg[\zeta_r]$. The other case then is: for 
some $1\le i\le (r-1)$, $J^{\sigma^i}=J$ and the restriction of $\sigma^i$ to $J$ is 
a semiregular automorphism of order $r$ of $J$, therefore $\dim_k J=
r\dim_k J_{\sigma^i}\geq r\dim_k e_J\Aalg=r\dim_k\Aalg$ (as $\sigma^i$ fixes $\Aalg$ 
it has to fix $e_J\Aalg$), which contradicts to $\dim_k J<\dim_k\Calg=r\dim_k \Aalg$. 
Secondly, we claim that for some $i\in \{1,\ldots,r-1\}$, $J\cap J^{\sigma^i}\not=0$. 
Indeed, assuming the contrary, we would have $J^{\sigma^j}\cap J^{\sigma^i}=
(J\cap J^{\sigma^{i-j}})^{\sigma^{j}}=0$ whenever $i\not\equiv j\pmod{r}$, whence the 
$J^{\sigma^i}$ would be pairwise orthogonal ideals, whence $\dim_k J=
{\frac{1}{r}}\dim_k\sum_{t=0}^{r-1}J^{\sigma^t}\leq{\frac{1}{r}}\dim_k\Calg=
\dim_k\Aalg$. This together with the fact that $\phi^{-1}$ injectively embeds $\Aalg$ 
in $J$ gives $J=\Aalg$, which implies that $\phi(\Calg)=\Aalg$, thus 
$\phi(\Aalg[\zeta_r][\sqrt[r]{c}])=\phi(\Calg[\zeta_r])=\Aalg[\zeta_r]$ contradicting 
$x\not\in\Aalg[\zeta_r]$.

{}From the above two claims we get an $i\in \{1,\ldots,r-1\}$, for which 
$J\not=J^{\sigma^i}$ and $J\cap J^{\sigma^i}\not=0$, whence by the method of Lemma 
\ref{lem-refine} we get a zero divisor of $J$, thus finishing the proof.
\end{proof}

Now we show the main result of this subsection: bringing down automorphisms of $\Dalg$
to $\Aalg\le\Dalg$.

\begin{proposition}\label{pro-baseaut}
Given a commutative semisimple algebra $\Dalg$ over a finite field $k$, its semiregular 
$k$-automorphism $\tau$ of prime order $r\not=char\ k$, a subalgebra $\Aalg\supset k$ 
of $\Dalg$ such that $\frac{\dim_k\Dalg}{\dim_k\Aalg}$ is an integer not divisible by 
$r$. Then we can find in deterministic $poly(\log |\Dalg|)$ time either a zero divisor 
in $\Aalg$ or a subalgebra $\Calg\le\Aalg$ together with a semiregular automorphism 
$\tau'$ of $\Calg$ of order $r$ such that $\Calg_{\tau'}\geq \Aalg_\tau 
(:=\Aalg\cap\Dalg_\tau)$.
\end{proposition}
\begin{proof}
We use the method of Proposition \ref{pro-teichresolv} to find an element 
$x\in T_{\Dalg,r}$ such that $x^\tau=\zeta_r x$. If $x\in\Aalg[\zeta_r]$ then we define 
$\Calg:=\Aalg_\tau[\zeta_r][x]_{\Delta_r}$. As $\tau$ fixes $\zeta_r$ while $\Delta_r$ 
fixes $\Dalg$, $\tau$ commutes with $\Delta_r$. Thus, $\Calg_\tau=
(\Aalg_\tau[\zeta_r][x]_\tau)_{\Delta_r}=\Aalg_\tau[\zeta_r]_{\Delta_r}=\Aalg_\tau$. 
This means that we have the $\Calg$ and the $\tau':=\tau|_\Calg$ as promised. On the 
other hand if $x\not\in\Aalg[\zeta_r]$ then we claim that we can find a zero divisor in 
$\Dalg$, decompose $\Dalg$ into a direct sum of orthogonal ideals and construct the 
$\Calg$ and the $\tau'$ in one of the ideals recursively.

Say $x\not\in\Aalg[\zeta_r]$, then since $x^{r^t}=1_\Dalg\in\Aalg$ for some integer 
$t>0$, we can choose a $y\in\{x,x^r,x^{r^2},\ldots\}$ such that  
$y\not\in\Aalg[\zeta_r]$ but $c^{\prime}:=y^r\in\Aalg[\zeta_r]$. By Lemma 
\ref{lem-primeext}, we can find a zero divisor in $\Dalg$ unless
$\Aalg[\zeta_r][\sqrt[r]{c^\prime}]$ is isomorphic to the subalgebra 
$\Aalg[\zeta_r][y]$. In the latter case 
$\Dalg_0:=\Aalg[\zeta_r][y]_{\Delta_r}\leq\Dalg$ is a free $\Aalg$-module of rank $r$, 
by Proposition \ref{pro-rootext}. Comparing dimensions it follows that $\Dalg$ cannot 
be a free $\Dalg_0$-module, therefore we can find a zero divisor $z$ in $\Dalg_0$ by 
Lemma \ref{lem-nonfree}. Thus, whenever $x\not\in\Aalg[\zeta_r]$, we can find a zero 
divisor $z$ in $\Dalg$. 

We proceed with computing the ideal of $\Dalg$ generated by $z$ and using Lemma 
\ref{lem-refine}, obtain a $\tau$-invariant decomposition of $\Dalg$ into the 
orthogonal ideals $I_1,\ldots,I_t$. For $1\le j\le t$, we denote by $\phi_j$ the 
projection $\Dalg\rightarrow I_j$. We can assume that for all $j$, $\phi_j|_\Aalg$ is 
injective as otherwise we find a zero divisor in $\Aalg$ and let 
$E\subseteq\{I_1,\ldots,I_t\}$ be a set of representatives of all the $r$-sized orbits 
of $\tau$. We have $\frac{\dim_k\Dalg}{\dim_k\Aalg}=
\sum_{j=1}^t\frac{\dim_k I_j}{\dim_k \Aalg}=
\sum_{I_j^\tau=I_j}\frac{\dim_k I_j}{\dim_k \Aalg}+
r\sum_{I_j\in E}\frac{\dim_k I_j}{\dim_k  \Aalg}$,
from which we infer that the first sum is nonempty and includes at least one term not 
divisible by $r$, therefore we can choose an index $j$ such that $I_j$ is 
$\tau$-invariant and $r\not|\frac{\dim_k I_j}{\dim_k \Aalg}$. So we can proceed with 
$I_j$ and $\phi_j\Aalg\cong\Aalg$ in place of $\Dalg$ and $\Aalg$ respectively in the 
algorithm described above. 

The process described above stops when either we find a zero divisor in $\Aalg$ or an
element $x\in T_{\Aalg',r}$ with $x^\tau=\zeta_r x$, where $\Aalg'\cong\Aalg$ is the 
image of $\Aalg$ under the projection $\phi$ of $\Dalg$ to some $\tau$-invariant ideal 
$I$. In the latter case we compute the subalgebra 
$\Calg':=\Aalg'_\tau[\zeta_r][x]_{\Delta_r}$. Finally put $\Calg:=\phi^{-1}(\Calg')$ 
and $\tau':=\phi^{-1}\circ\tau\circ\phi$. Notice that, if $e_I$ is the identity 
element of $I$ then $\tau$ will fix $e_I$ and $\phi:\Dalg\rightarrow I$ will just be 
the homomorphism $d\mapsto e_I d$, thus $\tau$ commutes with $\phi$. Consequently, 
$\Calg_{\tau'}=\phi^{-1}(\Calg'_\tau)=\phi^{-1}(\Aalg'_\tau)\geq \Aalg_\tau$.
\end{proof}

\subsection{Essential Part of the Tensor Power}\label{sec-essential}

Let $\Aalg$ be a commutative semisimple algebra over a finite field $k$. Let
$\Balg$ be its subalgebra such that $k\subseteq\Balg$ and $\Aalg$ be a free
module over $\Balg$ of rank $m$. If $char\ k\le m^2$ then polynomial factorization 
can be done in deterministic time by Berlekamp's algorithm and consequently,
all our results can be obtained easily. So we assume from now on that  
$char\ k>m^2$. But then we can also assume that $\Aalg$ is a simple extension 
algebra of $\Balg$ and find a primitive element $\alpha$ by running an 
algorithmic version of Fact \ref{fac-primitive} (if this ``fails'' then it gives a zero 
divisor of $\Aalg$). If $g(X)\in\Balg[X]$ is a minimal polynomial of $\alpha$ then 
we have that $\Aalg=\Balg[X]/(g(X))$.

It was shown by R\'onyai \cite{ro1} that, under GRH, a zero divisor in $\Aalg$ 
can be found in time $poly((\dim_k\Aalg)^r, \log |k|)$ if $r$ is a prime divisor 
of $\dim_k\Aalg$. In this section we extend the method of \cite{ro1} and obtain a GRH-free 
version that will be crucial in the proof of Main Theorem. A key idea of R\'onyai was
to work in the {\em essential part} of the tensor powers of $\Aalg$. Before going to the
formal definition of it we give a motivating definition assuming $\Aalg=k[X_1]/(f(X_1))$,
the essential part of $\Aalg^{\otimes_k 2}:=$ $\Aalg\otimes_k\Aalg$ is its ideal isomorphic 
to the algebra: 
$$k[X_1,X_2]/(f(X_1), f_2(X_1,X_2)),\ \text{~where } f_2(X_1,X_2):=
\frac{f(X_2)}{X_2-X_1}\in\Aalg[X_2].$$ 
Similarly, we can write down an expression for the essential part of 
$\Aalg^{\otimes_k r}$ inductively, as a factor algebra of $k[X_1,\ldots,X_r]$.

{\bf Functional interpretation of tensor powers: }
Let a commutative semisimple $\Aalg$ be a simple extension algebra over $\Balg\supseteq k$ 
such that $\Aalg=\Balg[X]/(g(X))$ and $g(X)\in\Balg[X]$ is 
a monic polynomial of degree $m$. Let $r\leq m$. We 
consider the $r$-th tensor power $\Aalg^{\otimes_\Balg r}$ ($\Aalg$ tensored with itself 
$r$ times wrt $\Balg$). To define (and compute) the essential part of this tensor power 
it is convenient to interpret $\Aalg$ as a collection of functions $V\rightarrow\overline
\Balg$ that are expressible as a polynomial over $\Balg$ (called {\em $\Balg$-polynomial 
functions}), where $\overline\Balg:=\overline k\otimes_k\Balg$ is the {\em algebraic closure} 
of $\Balg$ and  $V\subset\overline\Balg$ is a set of roots of $g(X)$. If 
$\Balg$ is not a field then there are various possibilities for $V$ and we need one with 
$\prod_{v\in V}(X-v)= g(X)$. Such a $V$ clearly exists by the definition of the algebraic 
closure. This {\em functional interpretation} of $\Aalg$ generalizes to 
$\Aalg\otimes_\Balg\Aalg$, which now becomes the set of all $\Balg$-polynomial functions 
from the set $V\times V$ to $\overline\Balg$ and finally $\Aalg^{\otimes_\Balg r}$ is the 
set of all $\Balg$-polynomial functions from the set $V^r$ to $\overline\Balg$. Note that 
in this interpretation a rank $1$ tensor element $h_1\otimes\cdots\otimes h_r$ in 
$\Aalg^{\otimes_\Balg r}$ corresponds to the function $V^r\rightarrow \overline\Balg$ 
that maps $(v_1,\ldots,v_r)\mapsto$ $h_1(v_1)\cdots h_r(v_r)$ .

{\bf Essential part of tensor powers: }
The {\em essential part} $\widetilde{\Aalg^{\otimes_\Balg r}}$ of 
$\Aalg^{\otimes_\Balg r}$ is the subset of functions that vanish on all 
the $r$-tuples $(v_1,\ldots,v_r)$ that have $v_i=v_j$ for some $i\neq j$. It can be seen that 
$\widetilde{\Aalg^{\otimes_\Balg r}}$ is an ideal of $\Aalg^{\otimes_\Balg r}$. We show below 
that given a basis of $\Aalg$ over $\Balg$ we can directly compute a basis for 
$\widetilde{\Aalg^{\otimes_\Balg r}}$ over $\Balg$.

\begin{lemma}\label{lem-essential-part}
A basis for $\widetilde{\Aalg^{\otimes_\Balg r}}$ over $\Balg$ can be computed 
by a deterministic algorithm in time $poly(m^r, \log |\Aalg|)$. 
\end{lemma}
\begin{proof}
Consider embeddings $\mu_i$ of $\Aalg$ into
$\Aalg^{\otimes_\Balg r}$ ($i=1,\ldots,r$) given as $\mu_i(a)=
1\otimes\ldots\otimes 1 \otimes a \otimes 1 \otimes \ldots\otimes 1$ where $a$ is 
in the $i$-th place. In the interpretation as functions,
$\mu_i(\Aalg)$ correspond to the $\Balg$-polynomial functions
on $V^r$ which depend only on the $i$th element in the tuples.
Observe that the set, for $1\le i< j\le r$:
$$\Delta^r_{i,j}=\{b\in \Aalg^{\otimes_\Balg r}\ |\ (\mu_i(a)-\mu_j(a))b=0
\forevery a\in \Aalg\}$$
is the ideal of $\Aalg^{\otimes_\Balg r}$ consisting of the $\Balg$-polynomial functions 
which are zero on every tuple $(v_1,\ldots,v_r)$ with $v_i\neq v_j$. Given a basis for
$\Aalg$, a basis for $\Delta^r_{i,j}$ can be computed
by solving a system of linear equations in time (counting $k$-operations as unit time)
polynomial in $\dim_k\Aalg^{\otimes_\Balg r}=m^r \dim_k\Balg$. 
Finally, notice that $\widetilde{\Aalg^{\otimes_\Balg r}}$ can be computed as well since 
it is the annihilator of $\sum_{1\le i<j\le r}\Delta^r_{i,j}$.
\end{proof}

{\bf Automorphisms of the essential part: }
The symmetric group $S_r$ acts as a group of automorphisms 
of $\Aalg^{\otimes_\Balg r}$. The action of $\pi\in S_r$ is 
the $\Balg$-linear extension of the map 
$h_1\otimes \cdots\otimes h_r
\mapsto h_{\pi(1)}\otimes \cdots\otimes h_{\pi(r)}$.
This action is not semiregular on the tensor power algebra as it fixes the set $I_0$ of
$\Balg$-polynomial functions on $V^r$ that are zero on all the points 
$V^r\setminus\{(v,\ldots,v)|v\in V\}$, where $I_0$ can be seen to be an ideal of   
$\Aalg^{\otimes_\Balg r}$. However, the ideal 
$\widetilde{\Aalg^{\otimes_\Balg r}}$
is invariant under this action and on it $S_r$ acts semiregularly. 

{\bf Embedding $\Aalg$ in the essential part: }
$\Aalg$ can be embedded into $\Aalg^{\otimes_\Balg r}$
by sending $h\in\Aalg$ to $h\otimes 1_\Aalg \otimes \cdots \otimes 1_\Aalg$.
Composing this embedding with the projection onto ideal $\widetilde{\Aalg^{\otimes_\Balg r}}$
(which exists by the semisimplicity of the tensor power) we obtain an embedding of 
$\Aalg$ in $\widetilde{\Aalg^{\otimes_\Balg r}}$. 

Note that the ideal $\widetilde{\Aalg^{\otimes_\Balg r}}$ is a free
$\Balg$-module of rank $m\cdots(m-r+1)$. 
Denoting the above embedding of $\Aalg$ also by $\Aalg$, if $r$ is a 
prime divisor of $m$ then $m\cdots(m-r+1)/m=$
$\dim_k {\widetilde{\Aalg^{\otimes_\Balg r}}}/\dim_k\Aalg $ is not divisible by $r$ and we can 
apply Proposition \ref{pro-baseaut} with $\widetilde{\Aalg^{\otimes_\Balg r}}$
as $\Dalg$ and the cyclic permutation $(1\ldots r)$ as $\tau$. This immediately gives us the 
following GRH-free version of the result of \cite{ro1}:

\begin{theorem}\label{thm-constructaut}
Let $\Balg$ be a subalgebra of a commutative semisimple algebra $\Aalg$ over a finite field 
$k$ such that $k\subseteq\Balg$; let $\Aalg$ be a free $\Balg$-module of rank $m$; and let 
$r$ be a prime divisor of $m$. Then in deterministic $poly(m^r, \log |\Aalg|)$ time one can
either find a zero divisor in $\Aalg$ or compute a subalgebra $\Calg$ of $\Aalg$ together with 
a semiregular automorphism $\tau$ of $\Calg$ of order $r$ such that $\Calg_\tau\ge\Balg$.
\end{theorem}

In the proof of Main Theorem we will need one more property of the essential part of the 
tensor square.

{\bf Left and Right Mappings: }
Note that there are two ways to map $\Aalg$ into an ideal $I\unlhd
\widetilde{\Aalg\otimes_\Balg\Aalg}$:
{\em either} by first embedding $\Aalg$ into $\Aalg\otimes_\Balg\Aalg$ by 
$h\mapsto h\otimes 1$ {\em or} by first embedding $\Aalg$ into $\Aalg\otimes_\Balg\Aalg$ 
by $h\mapsto 1\otimes h$, and then projecting to the ideal $I$ (which is also an ideal of 
$\Aalg\otimes_\Balg\Aalg$). The former we call the {\em left} mapping while the latter 
the {\em right} mapping (of $\Aalg$ into $I$). 

We will now show that these two mappings of $\Aalg$ into $I\unlhd
\widetilde{\Aalg\otimes_\Balg\Aalg}$ are quite different if $I$ is large enough.

\begin{lemma}\label{lem-left-right}
Let $m:=dim_\Balg\Aalg$ and $I$ be a nonzero ideal of 
$\widetilde{\Aalg\otimes_\Balg\Aalg}$.
Let $\tau_1:\Aalg\rightarrow I$ be the left mapping of $\Aalg$ while 
$\tau_2$ be the right mapping of $\Aalg$ into $I$. 
Then there exists an element $x\in\Aalg$ such that
$\tau_1(x)\neq \tau_2(x)$. Furthermore, if
$dim_k I/\dim_k\Balg >m$ then $\tau_1(\Aalg)\neq\tau_2(\Aalg)$. 
\end{lemma}
\begin{proof}
To see the first statement observe that
$\widetilde{\Aalg\otimes_\Balg\Aalg}$ is the ideal
of $\Aalg\otimes_\Balg \Aalg$ generated by the set of elements
$\{x\otimes 1-1\otimes x | x\in\Aalg\}$, see Lemma~\ref{lem-essential-part}. 
It follows that $I$ (as an ideal) is generated by the elements 
$\{\tau_1(x)-\tau_2(x) | x\in\Aalg\}$.
 Consequently, if $\tau_1(x)-\tau_2(x)=0$ for all $x\in\Aalg$ then $I=0$.

To see the second assertion, note that as $I$ is an ideal of the essential 
part of the semisimple $\Aalg\otimes_\Balg\Aalg$, there is a natural 
projection $\phi: \Aalg\otimes_\Balg \Aalg\rightarrow I$. Then 
$\tau_1(\Aalg)=\phi(\Aalg\otimes_\Balg 1)$ and $\tau_2(\Aalg)=
\phi(1\otimes_\Balg \Aalg)$.
{}From this and from the fact that $\Aalg\otimes_\Balg 1$ and $1\otimes_\Balg 
\Aalg$ generate $\Aalg\otimes_\Balg\Aalg$ we infer that
$\tau_1(\Aalg)$ and $\tau_2(\Aalg)$ generate $I$.
As $\dim_k\tau_i(\Aalg)\leq \dim_k\Aalg=m\dim_k\Balg<\dim_k I$,
this excludes the possibility of $\tau_1(\Aalg)=\tau_2(\Aalg)$.
\end{proof}
\rem{To see the second assertion, let $I^\perp$ be the ideal 
of $\Aalg\otimes_\Balg \Aalg$ complementary
to $I$. Then there is an isomorphism 
$\psi:(\Aalg\otimes_\Balg \Aalg)/I^\perp\cong I$
such that projection of $\Aalg\otimes_\Balg \Aalg$ to $I$ 
is the map $\psi\circ\phi$ where $\phi$ is the natural projection
$\Aalg\otimes_\Balg \Aalg\rightarrow (\Aalg\otimes_\Balg \Aalg)/I^\perp$.}
%

\subsection{Proof of Main Theorem}\label{Evd-sect}

We now prove the following slightly stronger version of Main Theorem.

\begin{theorem}\label{thm-Evditerate}
Given a commutative semisimple algebra $\Aalg$ over a finite field $k$ and a subalgebra 
$\Balg\supseteq k$ of $\Aalg$ such that $\Aalg$ is a free $\Balg$-module of rank $m$. 
Then in deterministic $poly(m^{\log m}, \log |\Aalg|)$ time one can either find a zero 
divisor in $\Aalg$ or a semiregular automorphism $\sigma$ of $\Aalg$ of order $m$ with 
$\Aalg_\sigma=\Balg$.
\end{theorem}
\begin{proof}
We may assume that $char\ k> m^2$ as otherwise using Berlekamp's factoring algorithm 
we can completely decompose $\Aalg$ into simple components.

If $m$ is even then using the algorithm of Theorem \ref{thm-constructaut} we either 
find a zero divisor in $\Aalg$ or a subalgebra $\Calg\le\Aalg$ together with a 
semiregular automorphism $\sigma_0$ of $\Calg$ of order $2$ with $\Calg_{\sigma_0}\ge
\Balg$ in deterministic polynomial time. In the former case we are done while in the 
latter case we make two recursive calls: one on the pair $(\Aalg, \Calg)$ and the other 
on the pair $(\Calg_{\sigma_0}, \Balg)$. This way we either find a zero divisor in 
$\Aalg$ or we find a semiregular automorphism $\sigma_1$ of $\Aalg$ satisfying 
$\Aalg_{\sigma_1}=\Calg$ as well as a semiregular automorphism $\sigma_2$ of 
$\Calg_{\sigma_0}$ satisfying $(\Calg_{\sigma_0})_{\sigma_2}=\Balg$. In the former case 
we are done while in the latter case we apply the algorithm of Lemma \ref{lem-extendaut} 
two times to construct $\sigma$ from $\sigma_0, \sigma_1, \sigma_2$. This finishes the 
even $m$ case.

Assume for the rest of the proof that $m$ is odd. We outline here the overall flow of the 
algorithm. We work in the algebra $\Aalg':=\widetilde {\Aalg{\otimes_\Balg
\Aalg}}$ and 
$\Balg':=\phi_1(\Aalg)$ where, $\phi_1$ and $\phi_2$ are respectively the left and right 
embeddings of $\Aalg$ into $\Aalg'$. During the course of the algorithm we maintain a 
nonzero ideal $I\unlhd\Aalg'$ with $\Balg'$ embedded in it. Any time we find a zero 
divisor in $I$ we replace $I$ with either the ideal generated by the zero divisor or its 
complement, depending on which has smaller dimension. We can assume the new ideal to be a 
free module over an embedded $\Balg'$ as otherwise we can find a zero divisor in $\Balg'$ 
(equivalently in $\Aalg$). Note that the rank of the new ideal over the embedded $\Balg'$ 
is at most half of the original one. Initially $I=\Aalg'$ and it is a free $\Balg'$-module 
of even rank $(m-1)$ and so we can apply the recursion outlined in the
second paragraph of this proof. In 
this way at any stage we either find a smaller ideal of $I$ or a semiregular automorphism 
$\sigma$ of $I$ such that $I_\sigma=e_I\Balg'\cong\Balg'$, where $e_I$ is the identity 
element of $I$. In the former case we replace $I$ by the smaller ideal (with an embedded 
$\Balg'$) and apply recursion which again either finds a zero divisor (and hence a smaller 
ideal) or a $\Balg'$-automorphism of the new ideal. 

The recursion outlined above halts either with a zero divisor found in $\Balg'$ (equivalently 
in $\Aalg$) or with a semiregular automorphism $\sigma$ of an $I\unlhd\Aalg'$ such that 
$I_\sigma=e_I\Balg'\cong\Balg'$. In the former case we are done while the latter case is what
we handle now. Let $\tau_1:\Aalg\rightarrow I$ mapping $a\mapsto e_I\phi_1(a)$ be the 
embedding of $\Aalg$ into $I$. Look at the homomorphism $\tau_2:\Aalg\rightarrow I$ that maps 
$a\mapsto e_I\phi_2(a)$. It is a nonzero homomorphism as $\tau_2(1)=e_I\not=0$. So we can 
assume $\tau_2$ to be an embedding of $\Aalg$ in $I$ as well or else we get a zero divisor in 
$\Aalg$ 

If $\sigma$ is trivial, i.e. $I=e_I\Balg'\cong\Balg'\cong\Aalg$, then $\mu:=\tau_2^{-1}\tau_1$ 
is a nontrivial $\Balg$-automorphism of $\Aalg$ by the first part of
Lemma~\ref{lem-left-right}. 
If $\mu$ is not semiregular then we can find 
a zero divisor by Proposition \ref{pro-regularaut} while if $\mu$ is semiregular then we can
apply recursion to the pair $(\Aalg_\mu,\Balg)$, find an automorphism of $\Aalg_\mu$ and 
finally extend it to a promised automorphism of $\Aalg$ by Lemma \ref{lem-extendaut}.  

So let us assume that $\sigma$ is nontrivial, i.e. $I>I_\sigma=\tau_1(\Aalg)$, thus 
$\rk_{\tau_1(\Balg)}I>m$. Then we define $\Balg'':=\tau_2(\Aalg)$ and apply recursion to the 
pair $(I,\Balg'')$. We either find a zero divisor of $I$ or obtain a semiregular automorphism 
$\sigma'$ of $I$ with $I_{\sigma'}=\Balg''$. In the former case we can proceed with a smaller 
ideal of $I$ or finish with a zero divisor of $\Balg''$ and hence of $\Aalg$, so the latter 
case of having a $\sigma'$ is what we think about now. We can assume that $\sigma$ and 
$\sigma'$ commute as otherwise we can find a zero divisor of $I$ by the algorithm of Theorem 
\ref{thm-cyclicaut} and proceed with recursion. Thus, $I_{\sigma'}$ is $\sigma$-invariant and 
$I_{\sigma}$ is $\sigma'$-invariant. Thus both $\sigma$ and $\sigma'$ can be viewed as 
automorphisms of $\tau_2(\Aalg)$ and $\tau_1(\Aalg)$ respectively. If both these actions are 
trivial then $\tau_1(\Aalg)=I_{\sigma}=(I_{\sigma})_{\sigma'}=(I_{\sigma'})_{\sigma}=
I_{\sigma'}=\tau_2(\Aalg)$, which contradicts the second statment of 
Lemma~\ref{lem-left-right}. Thus one of them is 
nontrivial, wlog say $\sigma$ is a nontrivial automorphism of $\tau_2(\Aalg)$. Then $\mu:=
\tau_2^{-1}\sigma\tau_2$ is a nontrivial automorphism of $\Aalg$. Again we can either find a 
zero divisor of $\Aalg$ or proceed with a recursion to the pair $(\Aalg_\mu,\Balg)$, getting 
a promised automorphism of $\Aalg$ by the algorithm of Lemma \ref{lem-extendaut}.

To see the dominating term in the time complexity observe that in any recursive call on some
pair, say $\Calg, \Dalg$ with $d:=\rk_\Dalg\Calg$, if $d$ is odd then we need to go to the
tensor square of $\Calg$ wrt $\Dalg$. Thus we need to then work in an algebra of 
rank
$d$ times the original rank. As we start with rank $m$ we have $d\le m$ and as 
the rank $d$ is at least halved in the subsequent recursive call (if there is one), we 
deduce that the algorithm works at all times in an algebra of rank (over $\Balg$) at 
most $m^{\log m}$. It is then routine to verify that the algorithm requires in all just 
$poly(m^{\log m})$ many $\Balg$-operations, which proves the time complexity as promised.
\end{proof}

To finish the proof of Main Theorem, apply the process described in the above 
Theorem to $\Balg=k$. If it yields a zero divisor $z$ of $\Aalg$ then the 
ideal $I:=\Aalg z$ and its complementary ideal $I^\perp$ give a decomposition
of $\Aalg=I\oplus I^\perp$. If $e_I$ is the identity element of $I$ then we 
can repeat the process now with $\Aalg$ replaced by $e_I\Aalg=I$ and $\Balg$ 
replaced by $e_I k\cong k$. Thus after several iterations based on Theorem 
\ref{thm-Evditerate} we get the direct sum decomposition of $\Aalg$ together 
with automorphisms as promised in Main Theorem.

\section{Noncommutative Applications}\label{sec-noncomm}

In this section we show that given a noncommutative algebra $\Aalg$ over a finite 
field we can unconditionally find zero divisors of $\Aalg$ in deterministic 
subexponential time. The idea is to compute a commutative subalgebra $\Dalg$ of $\Aalg$,
find an automorphism of $\Dalg$ using the algorithm described in Theorem 
\ref{thm-Evditerate}, and finally construct a zero divisor of $\Aalg$ using this 
automorphism.

{\bf Preprocessing: }
Let $\Aalg$ be a finite dimensional noncommutative algebra over a finite field $k$. If 
$\Aalg$ is not semisimple then we can compute the radical of $\Aalg$, by the 
deterministic polynomial time algorithm of \cite{ro3,CIW}, and get several zero divisors.
So we can assume that $\Aalg$ is semisimple. We can efficiently compute the center $\Calg$ 
of $\Aalg$ ($\Calg$ is the subalgebra having elements that commute with all elements in 
$\Aalg$) by solving a system of linear equations. By the Artin-Wedderburn Theorem (see 
Fact \ref{fac-wedderburn}) we know that if $\Calg_1,\ldots,\Calg_r$ are the simple 
components of $\Calg$ then, structurally, $\Aalg=\bigoplus_{i=1}^r M_{m_i}(\Calg_i)$,
where $M_m(R)$ stands for the algebra of all $m\times m$ matrices over the $k$-algebra 
$R$. Note that if the $m_i$'s are not all the same then $\Aalg$ would not be a free module
over $\Calg$ and hence we can find a zero divisor in $\Calg$ by Lemma \ref{lem-nonfree}.
So we can assume $\Aalg=\bigoplus_{i=1}^r M_m(\Calg_i)=M_m(\oplus_{i=1}^r\Calg_i)=
M_m(\Calg)$. Thus the hard case is to find a zero divisor in an algebra isomorphic to 
$M_m(\Calg)$, this is what we focus on in the remaining section. We identify $\Calg$ with 
the scalar matrices in $M_m(\Calg)$. 

\subsection{Automorphisms of a Commutative Semisimple Subalgebra of $M_m(\Calg)$}

Note that for any invertible matrix $A$ there is a natural automorphism of the full matrix
algebra that maps $x$ to $A^{-1}xA$, we call this a {\em conjugation} automorphism. We show 
in the first Lemma that, under certain mild condition,
an automorphism of a commutative semisimple subalgebra of the full matrix 
algebra corresponds to a conjugation automorphism. 

Recall that every maximal commutative semisimple algebra of the full matrix algebra
$M_m(F)$ over a perfect field $F$ has dimension $m$ over $F$. If $F$ is
algebraically closed then every commutative semisimple subalgebra of $M_m(F)$
is in fact (upto a conjugation isomorphism) a subalgebra of the diagonal matrices.

\begin{lemma}\label{lem-noetherskolem}
Let $\Calg$ be a commutative semisimple algebra over a finite field $k$, let $\Balg\le 
M_m(\Calg)$ be a commutative semisimple $\Calg$-algebra and let $\sigma$ be a 
$\Calg$-automorphism of $\Balg$. Let there be a maximal commutative semisimple subalgebra 
$\Dalg\le M_m(\Calg)$ containing $\Balg$ such that $\Dalg$ is a free $\Balg$-module. Then 
there exists a nonzero $y\in M_m(\Calg)$ such that $\forall x\in \Balg$, $x^\sigma=y^{-1}xy$.
\end{lemma}
\begin{proof}
We get hold of this element $y$ by reducing the question to the case of $\Calg$ being an 
algebraically closed field, when $\Dalg$ becomes a direct sum of $m$ copies of $\Calg$ and
$\Balg$ becomes a direct sum of $r|m$ copies of $\Calg$. In that case we can find a basis of 
$0$-$1$ diagonal matrices for $\Balg$ that is permuted by $\sigma$ and hence construct the 
promised $y$ as a permutation matrix.  

Firstly, we can assume $\Calg$ to be a field because if $I_1,\ldots,I_c$ are the simple components 
of $\Calg$ then clearly the $I_i$'s are all finite fields, and we can try finding the promised $y_i$ for
the instance of $(\Dalg I_i, \Balg I_i, I_i)$. Note that since $\sigma$ was fixing $I_i$, $\sigma$ is 
still a $(I_i)$-automorphism of $\Balg I_i$ and by freeness condition,
$\Dalg I_i$ is still a free 
$(\Balg I_i)$-module and it is a maximal commutative semisimple
subalgebra of $M_m(I_i)$.
 Also, once we have the $y_i$, for all $1\le i\le c$, satisfying $y_i x^\sigma=xy_i$
for all $x\in I_i$; it is easy to see that $(y_1+\ldots+y_r)$ is the promised $y$. So for the rest of the 
proof we assume that $\Calg$ is a finite field extension of $k$. Secondly, notice that the condition
$yx^\sigma=xy$ is equivalent to the system of equations: $yx_1^\sigma=x_1y, \ldots, 
yx_r^\sigma=x_r y$ for a $\Calg$-basis $x_1,\ldots,x_r$ of $\Balg$. In terms of the entries of the matrix 
$y$ this is a system of homogeneous linear equations in the field $\Calg$. This system has a nonzero 
solution over $\Calg$ iff the same system has a nonzero solution over the algebraic closure 
$\overline \Calg$ of $\Calg$. A solution over $\overline \Calg$ gives a matrix 
$y\in M_m({\overline \Calg})$ such that $yx^\sigma=xy$ for every $x\in {\overline \Balg}$ where
${\overline \Balg}:={\overline\Calg}\otimes_\Calg \Balg$ and we extend 
$\sigma$ ${\overline\Calg}$-linearly to an algebra automorphism of $\overline \Balg$. Because $k$ was a finite 
field, $\overline\Balg\le M_m({\overline \Calg})$ is a commutative semisimple algebra over 
$\overline\Calg$. Similarly, $\overline\Dalg:={\overline\Calg}\otimes_\Calg \Dalg$ is a maximal 
commutative semisimple subalgebra of $M_m({\overline \Calg})$, and is also a free 
$\overline \Balg$-module. By the former condition $\dim_{\overline \Calg}\overline\Dalg=m$ and by 
the latter condition $r|m$. We will now focus on the instance of $(\overline\Dalg, \overline\Balg, 
\overline\Calg)$ and try to construct the promised $y$.

As $\overline\Dalg$ is a sum of $m$ copies of $\overline\Calg$, by an appropriate basis change we 
can make $\overline\Dalg$ the algebra of all diagonal matrices in $M_m(\overline\Calg)$. Also, as 
$\overline\Dalg$ is a free $\overline\Balg$-module, a further basis change makes $\overline \Balg$ 
the algebra generated by the matrices $e_1,\ldots e_r$ where each $e_j$ is a diagonal $0$-$1$ 
matrix having $m/r$ consecutive $1$'s. In that case the automorphism $\sigma$ has a simple action, 
namely it permutes the matrices $\{e_1,\ldots,e_r\}$. Let $y$ be a block $r\times r$-matrix whose 
blocks are all $m/r\times m/r$ zero matrices except at positions $i,i^{\sigma}$ ($i^\sigma$ is defined 
by $e_i^\sigma=e_{i^{\sigma}}$), where the block is the $m/r\times m/r$ identity matrix. Clearly then, 
$e_{i^\sigma}=y^{-1}e_iy$ for all $1\le i \le r$ and hence $x^{\sigma}=y^{-1}xy$ for every 
$x\in\overline\Balg$ by extending the equalities linearly to $\overline\Balg$.
\end{proof}

In the second Lemma we show that a conjugation automorphism of prime order of a commutative 
semisimple subalgebra corresponds to a zero divisor of the original algebra.

\begin{lemma}\label{lem-minpolns}
Let $\Aalg$ be a finite dimensional algebra over the perfect field $F$ and let 
$\Balg\le\Aalg$ be a commutative semisimple algebra containing $F1_\Aalg$. 
Let $r$ be a 
prime different from $\chara F$ and let $y\in\Aalg$ be of order $r$ such that: 
$y^{-1}\Balg y=\Balg$ but there is an element $x\in\Balg$ with $y^{-1}xy\neq x$. 
Then the minimal polynomial of $y$ over $F$ is in fact $(X^r-1)$. 
As a consequence, $(y-1)$ and
$(1+y+\ldots+y^{r-1})$ is a pair of zero divisors in $\Aalg$.
\end{lemma}
\begin{proof}
Let $\overline F$ be the algebraic closure of $F$. Note that in $\overline\Aalg:=
{\overline F}\otimes_F \Aalg$, the 
minimal polynomial of $1\otimes y$ is the same as that of $y$ in $\Aalg$, 
$\overline\Balg:={\overline F}\otimes\Balg$ remains commutative semisimple and 
conjugation by $1\otimes y$ 
acts on it as an automorphism of order $r$. Thus for the rest of the proof we can assume 
$F$ to be algebraically closed.

As conjugation by $y$ does not fix $\Balg$, there exists a primitive idempotent
$e$ of  $\Balg$ for 
which the elements $e_j=y^{-j}ey^j$ ($j=1,\ldots,r$) are pairwise orthogonal primitive idempotents
of $\Balg$. This means that the corresponding left ideals $L_j:=\Aalg e_j$ are linearly independent 
over $F$. Assume now that the minimal polynomial of $y$ has degree less than $r$. So there are 
elements $\alpha_0,\ldots,\alpha_{r-1}\in F$, not all zero, such that $\sum_{j=0}^{r-1}\alpha_j y^j=0$. 
Implying that $e\sum_{j=0}^{r-1}\alpha_jy^j=\sum_{j=0}^{r-1}\alpha_jy^je_j=0$, this together 
with the fact that $y^je_j$'s are all nonzero, contradicts the linear independence of $L_1,\ldots,L_r$.
\end{proof}

\subsection{Proof of Application 1}

In this subsection we give the proof of Application 1: given a noncommutative algebra $\Aalg$ 
over a finite field $k$, one can unconditionally find zero divisors of $\Aalg$ in deterministic 
subexponential time. By the {\em preprocessing} discussed in the beginning of the section it is 
clear that we need to only handle the case of $\Aalg\cong M_m(\Calg)$, where $\Calg$ is a 
commutative semisimple algebra over $k$. The basic idea in the algorithm then is to find a 
maximal commutative semisimple subalgebra $\Dalg\le\Aalg$, find a $\Calg$-automorphism $\sigma$ 
of $\Dalg$, use it to define a subalgebra of $\Aalg$ which is a so called {\em cyclic algebra}, 
and then find a zero divisor in this cyclic algebra by the method of \cite{vdW}. 
The {\em cyclic algebras} $\Aalg'$ over $\Calg$ 
we encounter have two generators $x, y$ such that for a prime
$r$: $xy=\zeta_ryx$ and the multiplicative orders of $x, y$ are powers of $r$. These algebras
have the {\em ring of quaternions} as their classic special case, when $x^2=y^2=-1$ and $xy=-yx$.

\rem{We assume that $|k|>\dim_k\Aalg^2$.{\bf (Do we need this assumption??????)} }

Given the algebra $\Aalg$ (with an unknown isomorphism to $M_m(\Calg)$) in basis form over the 
finite field $k$. We can compute easily the center of $\Aalg$, and it will be $\Calg$. We can 
also compute a maximal commutative semisimple subalgebra $\Dalg$ of $\Aalg$ by the 
deterministic polynomial time algorithm of \cite{GI} ($\Dalg$ has an unknown isomorphism to the 
subalgebra of diagonal matrices of $M_m(\Calg)$). Being maximal, $\Dalg$ is a free module 
over $\Calg$ of rank $m$. By Theorem \ref{thm-Evditerate} we can, in deterministic 
$poly(m^{\log m}, \log |\Aalg|)$ time, either find a zero divisor in $\Dalg$ or compute a 
semiregular automorphism $\sigma$ of $\Dalg$ such that $\Dalg_\sigma=\Calg$. In the former 
case we are done, so it is the latter case that we now assume. By Lemma 
\ref{lem-noetherskolem}, there certainly exists a $y\in\Aalg$ such that $d^{\sigma}=y^{-1}dy$ 
for every $d\in\Dalg$, so by picking a nonzero solution of the corresponding system of linear 
equations we either find a zero divisor of $\Aalg$ or we find such a $y$. So suppose we find 
a $y$ such that $d^{\sigma}=y^{-1}dy\ne d$ for every $d\in\Dalg\setminus\Calg$.

We can efficiently obtain a multiple $M$ of the multiplicative order of $y$, $ord(y)$, just
by looking at the degrees of the irreducible factors of the minimal polynomial of $y$ over 
$k$ (this can be done deterministically without actually computing the factorization). Fix
a prime factor $r|m$, as $\sigma$ is a semiregular $\Calg$-automorphism of $\Dalg$, $\sigma$ 
is of order $m$, hence using $M$ we can replace $y$ and $\sigma$ by an appropriate power 
such that $ord(y)$ is a power of $r$ while $ord(\sigma)=r$. By this construction, conjugation 
by $y$ is now a $\Calg$-automorphism $\sigma$ of $\Dalg$ of order $r$. Put $z:=y^r$, thus 
$d=d^{\sigma^r}=z^{-1}dz$ for every $d\in\Dalg$. Note that we can assume $z\ne1$ as otherwise 
$(y-1)$ is a zero divisor of $\Aalg$ by Lemma \ref{lem-minpolns}. 
Thus an appropriate power, say $\zeta_r$, of $z$ has order $r$. Consider the subalgebra
$\Dalg[z]$, it is commutative by the action of $z$ on $\Dalg$ as seen before, it can 
also be assumed to be semisimple as otherwise we can find many zero divisors by just
computing its radical. So $\Dalg[z]$ is a commutative semisimple algebra. By the maximality
of $\Dalg$ we deduce that $\Dalg[z]=\Dalg$, hence $z\in\Dalg$ and $\zeta_r\in\Dalg$. So
by Lemma \ref{lem-auteigen} we can find efficiently either a zero divisor in $\Dalg$ or an 
$x\in \Dalg^*$ such that $x^\sigma=\zeta_r x$. We assume the latter case and we replace $x$ 
by an appropriate power so that $ord(x)$ is an $r$-power. Let $w:=x^r$, as $\sigma$ fixes 
$w$, it has to be in $\Calg$. 

Let $\Aalg':=\Calg[x,y]$, $\Dalg_x:=\Calg[x]\le\Aalg'$, $\Dalg_y:=\Calg[y]\le\Aalg'$ and $\Calg':=
\Calg[w,z]\le\Aalg'$. Note that by the definitions of $w, z$ it is easy to deduce that $\Calg'$ 
is in the center of $\Aalg'$ and $x, y\not\in\Calg'$. Furthermore by $xy=\zeta_r yx$ it follows that 
the set $\{x^iy^j | 1\le i,j\le (r-1)\}$ is a system of generators for $\Aalg'$ as a $\Calg'$-module. The relation 
$xy=\zeta_r yx$ also implies, that conjugation by $y$ acts on $\Dalg_x$ as an automorphism 
of order $r$ and that the conjugation by $x$ acts on $\Dalg_y$ as an automorphism of order 
$r$. We can assume that both these $\Calg'$-automorphisms are semiregular as otherwise we 
can find a 
zero divisor by Proposition \ref{pro-regularaut}. Thus both $\Dalg_x$ and $\Dalg_y$ are free 
modules over $\Calg$ of rank $r$, furthermore assume $\Aalg'$ to be a free $\Calg$-module 
(also free $\Calg'$-module) or else we find a zero divisor in $\Calg$ (or $\Calg'$) by Lemma 
\ref{lem-nonfree}.

We can assume that $w,z$ generate a cyclic subgroup of $\Calg'$ otherwise by Lemma 
\ref{lem-discretelog} we can find a zero divisor in $\Calg'$. If
the order of $z$ is larger than the order of $w$ then there is a $u\in\Calg'$ with $u^r=w$.
Put $x':=u^{-1}x$, then $x'^r=1$ and $x'y=\zeta_r yx'$, thus conjugation by $x'$ gives an
automorphism of $\Dalg_y$, whence $(x'-1)$ is a zero divisor
by Lemma \ref{lem-minpolns}. Similarly, we find a zero divisor if the order of $w$ is larger than 
the order of $z$. Thus we can assume that $w$ and $z$ have equal orders, say $r^t$.
By looking at the elements $w^{r^{t-1}}$ and $z^{r^{t-1}}$, both of which have order 
$r$ and they generate a cyclic group, we can find a unique $0<j<r$ such that 
$ord(w^jz)<r^t$. We now follow the method of the proof of Theorem 5.1 of \cite{vdW} to 
find a zero divisor in $\Aalg'$. 

Define $y':=x^jy$, and using $(yxy^{-1}=\zeta_r^{-1}x)$ repeatedly we 
get, $y'^r=(x^jy)^{r-2}(x^jy)(x^jy)$ $=(x^jy)^{r-3}(x^jy)(\zeta_r^{-j}x^{2j}y^2)=\cdots=
\zeta_r^{-jr(r-1)/2}x^{rj}y^r=\zeta_r^{-jr(r-1)/2}w^jz$. Thus if $r$ is odd then $y'^r=w^jz$,
and replacing $y$ with $y'$ leads to the case discussed above where the order of the new $z$ 
(i.e. $w^jz$) is less than that of $w$ (remember that $xy'=\zeta_r y'x$ still holds), and we 
already get a zero divisor. If $r=2$ then $y'^2=-wz$ ($j=1$), and the argument of the odd 
$r$ case can be repeated except when $ord(-wz)$ does not fall, i.e. orders are 
such that $ord(wz)<ord(w)=ord(z)=ord(-wz)$. This case is only possible (recall $z\ne1$) when 
$w=z=-1$, so $x^2=y^2=-1$ and $y^{-1}xy=-x$. Notice that in this case $\Aalg'$ is like a ring 
of quaternions and we handle this case next in a standard way.

To treat this case, by Theorem 6.1 of \cite{vdW}, one can efficiently find $\alpha,\beta\in k$ 
such that $\alpha^2+\beta^2=-1$.  Put $u:=(\alpha y+\beta)\in \Dalg_y$ and $x':=ux$. If 
$x'\in \Dalg_y$ then $x\in u^{-1}\Dalg_y=\Dalg_y$ which is a contradiction. Thus, 
$x'\not\in \Dalg_y$, in particular $x'\neq \pm 1$. While using $xy=-yx$ we can deduce that
$x'^2=(\alpha y+\beta)x(\alpha y+\beta)x=(\alpha y+\beta)(-\alpha y+\beta)x^2=
(\alpha^2+\beta^2)(-1)=1$. Thus $(x'-1)$ is a zero divisor. This finishes the proof of 
Application 1 in all cases.

\subsection{Further Results on Finding Zero Divisors in $M_m(\Calg)$}

In this part we briefly outline an alternative
of the approach of Application~1. Formal statements
and details of proofs will be subject of a
subsequent paper.

Assume that $\Aalg\cong M_m(\Calg)$ for some
commutative semisimple algebra $\Calg$ over the finite field
$k$. As in the proof of Application 1, we use 
the method of \cite{GI} to find a maximal
semisimple subalgebra $\Dalg$ of $\Aalg$. Note
that $\Dalg$ is a free module over $\Calg$ of rank
$m$. Let $r$ be a prime divisor of $m$. Then
we can use the algorithm of
Theorem~\ref{thm-constructaut} to
find an automorphism of a {\em subalgebra}
$\Balg$ of order $r$ in time
$poly(m^r, \log |\Aalg|)$. The remaining
part of the proof of Application 1 can
be modified so that an automorphism
of prime order of a subalgebra of $\Dalg$
rather than one of the whole $\Dalg$
can be used to find a zero divisor in $\Aalg$
in polynomial time. This way
we obtain a deterministic algorithm 
of complexity $poly(m^r, \log |\Aalg|)$
for finding a zero divisor in an
algebra $\Aalg$ isomorphic to $M_m(\Calg)$,
where $r$ is the smallest prime divisor of $m$.

Using a generalization \cite{ChIK} of 
a method of \cite{BR} we can use the
zero divisor obtained above to compute
a subalgebra of $\Aalg$ (in the broader sense, thus
a subalgebra of a one-sided ideal of $\Aalg$) isomorphic
to $M_{m'}(\Calg)$, where $m'$ is a certain divisor of $m$.
Iterating this method we ultimately find
a zero divisor $z$ of $\Aalg$ which is
equivalent to an elementary matrix (a matrix
having just one nonzero entry)
wrt an isomorphism $\Aalg\cong M_m(\Calg)$.
Then the left ideal $\Aalg z$ is isomorphic 
to the standard module for $M_m(\Calg)$ (the
module of column vectors of length $m$ over $\Calg$).
Finding such a module is equivalent to
constructing an explicit isomorphism
with $M_m(\Calg)$. The time complexity
is $poly(m^r, \log |\Aalg|)$,
where $r$ is the {\em largest} prime divisor of $m$.
In particular, if $\Aalg\cong M_{2^\ell}(\Calg)$,
our method computes such an isomorphism in
deterministic {\em polynomial time}.

\section{Special Finite Fields: Proof of Application 4}\label{sec-p-1}

In this section we assume that $k=\F_p$ for a prime $p>3$
and the prime factors of $(p-1)$ are bounded by $S$.
We also assume that all the algebras that appear in the section are 
completely {\em split} semisimple algebras over $k$,
i.e. isomorphic to direct sums of copies of $k$.

We first show an algorithm that constructs an $r$-th Kummer extension of an 
algebra given a prime $r|(p-1)$. We basically generalize Lemma~2.3~of~\cite{ro2} 
to the following form:

\begin{lemma}\label{lem-smooth1}
Assume that $\Aalg$ is a free module over its subalgebra
$\Balg$ of rank $d$. Then in time $poly(\log|\Aalg|,S)$
we can find either a zero divisor in $\Aalg$ or an element $x\in \Aalg^*$ with a 
power of $r$ order, for a prime $r|(p-1)$, satisfying one of the following 
conditions:
\\(1) $r\neq d$, $x\not\in\Balg$ and $x^r\in\Balg$,
\\(2) $r=d$, $x^r\not\in\Balg$ and $x^{r^2}\in\Balg$,
\end{lemma}
\begin{proof}
As $\Balg$ is a completely split semisimple algebra, say of dimension $n$ 
over $k$, there are orthogonal primitive idempotents $f_1,\ldots,f_n$ such 
that $f_i\Balg\cong k$ for all $i$. For an $i\in\{1,\ldots,n\}$, we can 
project the hypothesis to the $f_i$ component, thus $\dim_k f_i\Aalg=d$ and 
there are orthogonal primitive idempotents $e_{i,1},\ldots,e_{i,d}$ of 
$\Aalg$ such that $f_i\Aalg=e_{i,1}\Aalg\oplus\cdots\oplus e_{i,d}\Aalg$. 
As $f_i$ is an identity element of $f_i\Aalg$ we further get that 
$f_i=(e_{i,1}+\cdots+e_{i,d})$.

Now pick an $y\in\Aalg\setminus\Balg$. Suppose (for the sake of contradiction) 
for all $1\le i\le n$ there is a single $y_i^*\in k$ that satisfies for all 
$1\le j\le d$, $ye_{i,j}=y_i^* e_{i,j}$. Then their sum gives us that 
$y=\sum_{i=1}^n y_i^*f_i$, as each $y_i^*f_i\in\Balg$ we further get that 
$y\in\Balg$. This contradiction shows that there is an $i\in\{1,\ldots,n\}$ and
distinct $j,j'\in\{1,\ldots,d\}$ such that $ye_{i,j}=y_1 e_{i,j}$ and 
$ye_{i,j'}=y_2 e_{i,j'}$ for some $y_1\ne y_2\in k$. Let us fix these 
$i,j,j',y_1,y_2$ for the rest of the proof, we do not compute them but use their 
existence for the correctness of the algorithm. We can assume $y\in\Aalg^*$ 
otherwise we have a zero divisor and we are done.

Let $r_1,\ldots,r_t$ be the prime divisors of $(p-1)$. Let us assume 
$p\geq (S\log p+1)$ as otherwise we can just invoke Berlekamp's polynomial 
factoring algorithm to find a complete split of $\Aalg$, and we are done.
As $p\geq (S\log p+1)$ then 
there is an integer $0\leq a<(S\log p+1)$ such that $(y_1+a)^{r_\ell}\ne
(y_2+a)^{r_\ell}$ for all $\ell\in\{1,\ldots,t\}$ (since there can be at most $tS$ 
elements in $\F_p$ satisfying at least one of these equations). We could also assume 
$(y+a)$ to be invertible as otherwise we are done. Note that $(y+a)^{r_\ell}e_{i,j}=
(y_1+a)^{r_\ell}e_{i,j}$ and $(y+a)^{r_\ell}e_{i,j'}=(y_2+a)^{r_\ell} e_{i,j'}$ which 
together with $(y_1+a)^{r_\ell}\ne(y_2+a)^{r_\ell}$ implies that
$(y+a)^{r_\ell}\not\in\Balg$. Thus $z:=(y+a)$ is an element in $\Aalg^*$ for which  
$z^{r_\ell}\not\in\Balg$ for $\ell\in\{1,\ldots,t\}$.

Note that $z^{p-1}=1$, in particular $z^{p-1}\in\Balg$. Thus we can find two, not 
necessarily distinct, prime divisors $r_1$ and $r_2$ of $(p-1)$ such that replacing 
$z$ with an appropriate power of it we have $z^{r_1},z^{r_2}\not\in\Balg$
but $z^{r_1r_2}\in\Balg$. Either $r_1=r_2=d$ and we take $(x,r)=(z,d)$, or 
$r_1\neq r_2$ in which case say wlog $r_1\neq d$ and we take $(x,r)=(z^{r_2},r_1)$.
Finally we can raise $x$ by a suitable power (coprime to $r$) so that $x$ has a 
power of $r$ order together with the other properties.
\end{proof}

For an integer $m$ we denote by $\Phi_m(X)$ the $m$th cyclotomic polynomial in 
$k[X]$. Let $r_1,\ldots,r_t$ be the prime divisors of $(p-1)$. Then for a subset $I$ 
of $\{1,\ldots,t\}$ we denote the product $\prod_{i\in I}r_i$ by $r_I$.
We now give an algorithm that either finds a zero divisor in $\Aalg$ or a homomorphism
from an $r_I$-th cyclotomic extension onto $\Aalg$.

\begin{lemma}\label{lem-smooth2}
Let $\Balg<\Aalg$. Assume that we are also given
a surjective homomorphism from $k[X]/(\Phi_{r_I}(X))$
onto $\Balg$ for some subset $I$ of $\{1,\ldots,t\}$.
Then in time $poly(\log|\Aalg|, S)$ we can compute
either a zero divisor in $\Aalg$ or a subalgebra $\Balg'>\Balg$ of
$\Aalg$ together with a surjective homomorphism from 
$k[X]/(\Phi_{r_{I'}}(X))$ onto $\Balg'$ for some subset 
$I'\subseteq \{1,\ldots,t\}$.
\end{lemma}
\begin{proof}
We may clearly assume that $\Aalg$ is a free module
(of rank $d$) over $\Balg$.
Let the prime $r$ and the element $x\in\Aalg^*$ be the result of
an application of the algorithm of Lemma \ref{lem-smooth1}.
If $\Balg[x]$ is a proper subalgebra of $\Aalg$ then we can solve the 
problem by two recursive calls: first on $(\Balg[x],\Balg)$ and then 
on $(\Aalg,\Balg[x])$. Thus the base case of the recursion is when 
$\Aalg=\Balg[x]$. We handle this case now. In this case clearly $d\le r$. 

Assume case (2) i.e. $d=r$. We can assume $\Aalg=\Balg[x^r]$ as otherwise 
the subalgebra $\Balg[x^r]$ is a proper subalgebra of $\Aalg$ and we can 
find a zero divisor because $\Aalg$ cannot be a free module over this
subalgebra (as $dim_\Balg\Aalg=r$ is a prime). It follows that 
$\Phi_r(x^r)\neq 0$ because otherwise the rank of $\Aalg$ as a $\Balg$-module
would be at most $\phi(r)<r$, a contradiction. {}So we can assume $x^{r^2}\ne1$ 
as otherwise $\Phi_r(x^r)|(x^{r^2}-1)$ is a zero divisor and we are done. 
Thus we can find a power $\zeta\neq 1$ of
$x^{r^2}$ for which $\zeta^r=1$. This means, in particular, that a primitive 
$r$-th root of unity is in $\Balg$, and we have $\Aalg\cong\Balg[X]/(X^r-x^{r^2})$.
So we get a $\Balg$-automorphism $\sigma$ of $\Aalg$ that sends 
$x^r\mapsto\zeta x^r$. 
The automorphism $\sigma$ is of order $r$, is semiregular and satisfies
$\Aalg_\sigma=\Balg$. We compute the element
$z:=\prod_{i=0}^{r-1} x^{\sigma^i}$. Then $z^\sigma=z$, therefore $z\in\Balg$. Also,
$z^r=\prod_{i=0}^{r-1} {(x^r)}^{\sigma^i}=\zeta^{r(r-1)/2}x^{r^2}$.
If $r$ is odd then $z^r=x^{r^2}$ while $z\ne \zeta^i x^r$ for all $i$ 
($z, \zeta^i\in\Balg$ but $x^r\not\in\Balg$), thus $(z-\zeta^ix^r)$ is a zero divisor of 
$\Aalg$, for some $i$, and we are done. If $r=2$ then $z^2=-x^4$. We use the algorithm 
of \cite{Schoof} 
for finding a square root $w$ of $-1$ in $k$, observe that $(wz)^2=x^4$. Again as 
$wz\ne\pm x^2$ ($z, w\in\Balg$ but $x^2\not\in\Balg$), thus $(wz-x^2)$ is a zero divisor 
of $\Aalg$ and we are done.

Assume case (1) i.e. $d<r$, with $x^r\ne1$. We could assume $\Aalg=\Balg[x]$ to be a free
$\Balg$-module with the free basis $\{1,x,\ldots,x^{d-1}\}$, as otherwise we can find a 
zero divisor in $\Balg$ by Lemma \ref{lem-nonfree}. Also we can find a power $\zeta\neq 1$ 
of $x^r$ for which $\zeta^r=1$. These two facts mean that there is a well defined 
endomorphism $\phi$ of $\Aalg$ that maps $x$ to $\zeta x$ and fixes $\Balg$. Compute the 
kernel $J\subsetneq\Aalg$ of this endomorphism. If $J$ is nonzero then the elements of $J$
are zero divisors of $\Aalg$ (as $\phi$ cannot send a unit to zero), and we are done.
If $J$ is zero then $\phi$ is a $\Balg$-automorphism of $\Aalg$, clearly of order $r$. As 
$\dim_\Balg\Aalg<r$, $\phi$ cannot be semiregular, so we get a zero divisor by 
Proposition \ref{pro-regularaut} and we are done.

Finally assume again case (1) i.e. $d<r$, with $x^r=1$.
Let $\psi$ denote the given map $k[X]/(\Phi_{r_{I}}(X))$
onto $\Balg$. If $r\in I$ then put $y:=\psi(X^{r_I/r})$. Then $y\in\Balg^*\setminus\{1\}$
because $X^{r_I/r}, (X^{r_I/r}-1)$ are coprime to $\Phi_{r_I}(X)$ and are thus units. 
As $x^r=y^r$ but $x\ne x^i y$ for all $i$ ($y\in\Balg$ while $x\not\in\Balg$), we deduce
that $(x-x^i y)$ is a zero divisor for some $i$, and we are done.
Assume that $r\not\in I$. Let $I':=I\cup\{r\}$ 
and let $\Calg=k[X]/(\Phi_{r_{I'}}(X))$. We now break $\Calg$ using Chinese Remaindering.
Let $q_1$ be a multiple of $r$ which is congruent to 1 modulo $r_I$ and let $q_2$ be
a multiple of $r_I$ congruent 1 modulo $r$. Let
$X_1:=X^{q_1}$, $X_2:=X^{q_2}$ and let $\Calg_1$ resp.~$\Calg_2$ be the subalgebras
of $\Calg$ generated by $X_1$ resp.~$X_2$. Then
 $\Calg_1\cong k[X_1]/(\Phi_{r_{I}}(X_1))$ and
 $\Calg_2\cong k[X_2]/(\Phi_{r}(X_2))$. Let $\psi_1$ be the given
surjective map from $\Calg_1$ onto $\Balg$ 
and let $\psi_2$ be the map from $\Calg_2$ sending $X_2$ to
$x$. Let $\psi'$ be the map from $\Calg\cong\Calg_1\oplus\Calg_2$ into $\Aalg$
that is the linear extension of the map sending $X^i=(X_1^i,X_2^i)$ to 
$\psi_1(X_1^i)\psi_2(X_2^i)$. Clearly, $\psi'$ 
is a homomorphism from $\Calg$ to $\Aalg$ and is onto (as $\Aalg=\Balg[x]$).
This finishes the proof.  
\end{proof}

Using Lemma \ref{lem-smooth2} as an induction tool,
we obtain the following.

\begin{theorem}\label{thm-smooth}
Let $f(X)$ be a polynomial of degree $n$ which completely splits into
linear factors over $\F_p$. Let $r_1<\ldots<r_t$ be the prime factors
of $(p-1)$. Then by a deterministic algorithm of running time
$poly(r_t,n,\log p)$, we can either
find a nontrivial factor of $f(X)$ or compute a surjective homomorphism
$\psi$ from $\F_p[X]/(\Phi_{r_I}[X])$ to
$\F_p[X]/(f(X))$, where $r_I=\prod_{i\in I}r_i$ for some
subset $I$ of $\{1,\ldots,t\}$ and
$\Phi_{r_I}(X)$ is the cyclotomic polynomial of degree
$\prod_{i\in I}(r_i-1)$.
\end{theorem}
\qed

Note that if $\psi$ is not an isomorphism then we can break the cyclotomic
ring above and find its invariant decomposition into ideals by Lemma \ref{lem-refine}. 
As we know the automorphism group of cyclotomic extension rings over $\F_p$
(and of their ideals as well), this theorem immediately implies the statement 
of Application 4. 

\fi

\end{document}